%% file: main.tex
\documentclass[11pt]{article}

\usepackage{fullpage}
\usepackage{graphicx}
\usepackage{ifthen}
\usepackage{amsmath}
\usepackage{amsthm}
\usepackage{amssymb}
\usepackage{natbib}
\usepackage{mathtools}
\usepackage{hyperref}

\usepackage{nicefrac}

\usepackage{enumerate}

\usepackage[noabbrev,nameinlink]{cleveref}

\usepackage[dvipsnames]{xcolor}

\usepackage{tikz}

\newcommand{\ylabel}[2]{%
\begin{scope}[shift={(0,#1)}]
\draw (-.5ex,0) -- (.5ex,0);%
\draw[anchor=mid east] (-.5ex, 0) node {#2};%
\end{scope}
}

\newcommand{\xlabel}[2]{%
\begin{scope}[shift={(#1,0)}]
\draw (0, .5ex) -- (0, -.5ex); 
\draw[anchor=mid] (0, -2ex) node {#2};
\end{scope}
}

\input{macros}

\input{notation}

\input{math}

\begin{document}

%\begin{frontmatter}


\title{Dashboard Mechanisms for Online Marketplaces\footnote{A two-page abstract appeared in the 2019 ACM Conference on Economics and Computation.  The authors wish to thank Chenhao Zhang for the observation of \Cref{l:p0}.}}
\date{}

\newcommand{\email}[1]{\href{mailto:#1}{#1}}

%\author{ec 2019 Submission 375}
\author{Jason Hartline\thanks{Northwestern U., Evanston IL.  Work done in part while authors were visiting Booking.com and in part while supported by NSF CCF 1618502. Email: \email{hartline@northwestern.edu}, \email{aleckjohnsen2012@u.northwestern.edu}}
\and Aleck Johnsen\footnotemark[2]
\and Denis Nekipelov\thanks{U. of Virginia, Charlottesville, VA.  Work done in part while author was visiting Booking.com and in part while supported by NSF CCF 1563708.  Email: \email{denis@virginia.edu}}  
\and Onno Zoeter\thanks{Booking.com, Amsterdam, Netherlands.  Email: \email{onno.zoeter@booking.com}}}

\maketitle

\input{abstract}

\input{intro}

\input{related_work}

\input{overview}

\input{prelim}

\input{single-agent}

\input{dashboards}

\input{inferred-values-dashboard}

\input{rebalancing}

\input{minimal-rebalancing}

\input{blackbox}

\bibliographystyle{apalike}
\bibliography{bib}

%\begin{appendix}

%\input{a_proofs}
%\end{appendix}

%\input{persuadable}

%\input{fictitious}

\end{document}

%% file: macros.tex
\usepackage{amsmath,amsfonts,amssymb,amsthm,latexsym}
\usepackage{graphicx}
\usepackage{color}
\usepackage{bm}
\newtheorem{theorem}{Theorem}
\newtheorem{definition}{Definition}
\newtheorem{lemma}{Lemma}
\newtheorem{prop}{Proposition}
\newtheorem{fact}{Fact}
\newtheorem{corollary}{Corollary}

%% file: notation.tex
\usepackage{ifthen}
\usepackage{amssymb}
\usepackage{amsfonts}

\usepackage{bbm}
\usepackage{BOONDOX-cal}
\newcommand{\composed}[3]{#1{#2{#3}}}
\newcommand{\double}[2]{\kern.15ex#1{\kern-.15ex#1{\kern-.15ex#2}}}

\newcommand{\themodifier}{}  % '','forval','forquant','forbid','ironed','constrained'

\newcommand{\ironed}{\bar}
\newcommand{\constrained}{\hat}
\newcommand{\random}{\check}

\newcommand{\thestage}{an}           % 'an','exante', 'expost','interim'

\newcommand{\interimfont}[1]{\mathnormal{#1}}

\newcommand{\interim}[1]{{\setstagefont{\interimfont}#1}}

\newcommand{\stagefont}{\mathrm}
\newcommand{\setstagefont}[1]{\renewcommand{\stagefont}{#1}}

\newcommand{\decoration}{\noaccents}

\newcommand{\stagemodifier}[2]{\composed%
{\ifthenelse{\equal{#2}{}}{}{\csname #2\endcsname}}%
{\ifthenelse{\equal{#1}{}}{}{\csname #1\endcsname}}%
}

\newcommand{\thestagemodifier}{\stagemodifier{\thestage}{\themodifier}}

\newcommand{\ensuredecoration}{\renewcommand{\decoration}{\thestagemodifier}}

\ensuredecoration

\newcommand{\optconstrained}{\composed{\optimized}{\constrained}}
\newcommand{\optimized}[1]{#1\opt}
\newcommand{\differentiated}[1]{#1'}
\newcommand{\tagged}[2]{{#2}^{#1}}

\newcommand{\primedarg}[1]{#1\primed}
\newcommand{\noaccents}[1]{#1}

\newcommand{\drop}[1]{}
\newcommand{\definevariant}[3]{\expandafter\newcommand\expandafter{\csname \expandafter\drop\string#1#2\endcsname}{}
\expandafter\DeclareRobustCommand\expandafter{\csname \expandafter\drop\string#1#2\endcsname}{#3}}
\newcommand{\usevariant}[2]{\csname \expandafter\drop\string#1#2\endcsname}

\newcommand{\agind}[1][\agent]{_{#1}}
\newcommand{\agith}[1][\agent]{_{(#1)}}
\newcommand{\minusagind}[1][\agent]{_{-#1}}

\newcommand{\newagentvar}[3][\decoration]{%
\definevariant{#2}{}{#1{\stagefont{#3}}}%
\definevariant{#2}{i}{#1{\stagefont{#3}}\agind}
\definevariant{#2}{ith}{#1{\stagefont{#3}}\agith}
\definevariant{#2}{s}{\kern.1ex#1{\kern-.1ex\boldsymbol{\stagefont{#3}}}}
\definevariant{#2}{smi}{\kern.1ex#1{\kern-.1ex\boldsymbol{\stagefont{#3}}}\minusagind} 
}%

\newcommand{\itind}[1][\itm]{_{#1}}
\newcommand{\minusitind}[1][\itm]{_{-#1}}

\newcommand{\newitemvar}[3][\decoration]{%
\definevariant{#2}{}{#1{\stagefont{#3}}}%
\definevariant{#2}{j}{#1{\stagefont{#3}}\itind}
\definevariant{#2}{s}{\kern.1ex#1{\kern-.1ex\boldsymbol{\stagefont{#3}}}}
\definevariant{#2}{smi}{\kern.1ex#1{\kern-.1ex\boldsymbol{\stagefont{#3}}}\minusitind} 
}%

\newcommand{\newagentvarsm}[5][]{\newagentvar[\composed{#1}{\stagemodifier{#4}{#5}}]{#2}{#3}}

\newcommand{\theval}{v}
\newagentvar{\val}{\theval}
\newagentvar[\interim]{\Val}{\theval}
\newagentvar{\type}{t}
\newagentvar{\othertype}{s}

\newagentvarsm{\rev}{R}{interim}{nomodifer}
\newagentvar[\differentiated]{\marg}{\rev}
\newagentvarsm{\rawrev}{P}{interim}{}
\newagentvar[\differentiated]{\rawmarg}{\rawrev}

\newagentvarsm{\virt}{\phi}{interim}{}
\newagentvar{\cumvirt}{\Phi}
\newagentvar{\qvirt}{\phi}
\newagentvar[\ironed]{\ivirt}{\virt}
\newagentvar[\ironed]{\icumvirt}{\cumvirt}

\newagentvarsm{\dist}{F}{interim}{}
\newagentvarsm{\dens}{f}{interim}{}
\newagentvarsm{\hazard}{h}{interim}{}
\newagentvarsm{\cumhazard}{H}{interim}{}

\newcommand{\thealloc}{x}
\newagentvar{\alloc}{\thealloc}
\newagentvar[\ironed]{\ialloc}{\thealloc}

\newagentvarsm{\zee}{z}{an}{}
\newagentvarsm{\wye}{y}{an}{}

\newcommand{\thefalloc}{z}
\newagentvar{\falloc}{\thefalloc}
\newagentvarsm{\epfalloc}{\thefalloc}{expost}{}

\newcommand{\thequant}{q}
\newagentvar{\quant}{\thequant}
\newagentvar[\constrained]{\exquant}{\thequant}
\newagentvar[\constrained]{\critquant}{\thequant}  %remove?
\newagentvar[\optconstrained]{\monoq}{\thequant}
\newagentvar[\interim]{\toquant}{\thequant}

\newcommand{\theqalloc}{y}
\newcommand{\thecumalloc}{Y}
\newagentvar{\qalloc}{\theqalloc}
\newagentvar{\cumalloc}{\thecumalloc}
\newagentvarsm{\calloc}{\theqalloc}{interim}{}
\newagentvarsm{\cumcalloc}{\thecumalloc}{interim}{}
\newagentvarsm{\iqalloc}{\theqalloc}{interim}{ironed}
\newagentvarsm{\icumalloc}{\thecumalloc}{interim}{ironed}

\newagentvar{\qprice}{\price}
\newagentvar{\qrev}{R}
\newagentvar[\ironed]{\iqrev}{\qrev}

\newagentvarsm{\excalloc}{\theqalloc}{expost}{constrained}
\newagentvarsm{\exalloc}{\theqalloc}{expost}{forquantile}
\newagentvarsm{\extalloc}{\thealloc}{expost}{forvalue}
\newagentvarsm{\exfalloc}{\thefalloc}{expost}{forvalue}

\newagentvar[\primedarg]{\inducedrev}{\rev}

\newagentvar[\tilde]{\pseudorev}{\rev}
\newagentvar[\tilde]{\pseudorawrev}{\rawrev}

\newagentvar{\typespace}{{\cal T}}
\newagentvar{\typesubspace}{S}

\newagentvar{\outcome}{w}
\newagentvar{\outcomespace}{{\cal W}}

\newcommand{\served}[1]{#1^1}
\newcommand{\nonserved}[1]{#1^0}
\newcommand{\alloced}[1]{#1^{\alloc}}
\newcommand{\allocedi}[1]{#1^{\alloci}}
\newagentvar[\alloced]{\xoutcome}{\outcome}
\newagentvar[\allocedi]{\xioutcome}{\outcome}
\newagentvar[\served]{\soutcome}{\outcome}

\newagentvar[\nonserved]{\nsoutcome}{\outcome}

\newcommand{\theprice}{p}
\newagentvar{\price}{\theprice}

\newcommand{\thepay}{p}
\newagentvarsm{\pay}{\thepay}{an}{}
\newagentvarsm{\vpay}{\thepay}{interim}{forval}
\newagentvarsm{\epvpay}{\thepay}{expost}{forval}
\newagentvarsm{\talloc}{\thealloc}{interim}{forval}
\newagentvarsm{\valloc}{\thealloc}{interim}{forval}
\newagentvarsm{\eptalloc}{\thealloc}{expost}{forval}
\newagentvarsm{\epvalloc}{\thealloc}{expost}{forval}

\newagentvarsm{\eval}{\theval}{an}{constrained}
\newcommand{\thedalloc}{y}
\newagentvarsm{\dalloc}{\thedalloc}{interim}{forbid}
\newagentvarsm{\ealloc}{\thedalloc}{interim}{forval}
\newcommand{\thedstrat}{c}
\newagentvarsm{\dstrat}{\thedstrat}{interim}{forval}
\newcommand{\theepay}{q}
\newagentvarsm{\epay}{\theepay}{interim}{forval}

\newagentvarsm{\balloc}{\thealloc}{interim}{forbid}
\newagentvarsm{\bpay}{\thepay}{interim}{forbid}
\newagentvarsm{\epballoc}{\thealloc}{expost}{forbid}
\newagentvarsm{\epbpay}{\thepay}{expost}{forbid}

\newagentvar{\act}{a}
\newagentvar{\bidspace}{A}
\newagentvar{\actspace}{A}

\newagentvarsm{\critval}{\theval}{an}{constrained}
\newagentvarsm{\epvcritval}{\theval}{expost}{constrained}

\newagentvar[\constrained]{\crittype}{\type}
\newagentvar[\constrained]{\critvirt}{\virt}
\newagentvar[\constrained]{\reserve}{\val} % should this be used???
\newagentvar[\constrained]{\bidreserve}{\bid} % should this be used???
\newagentvar[\optconstrained]{\monop}{\val}
\newagentvar[\constrained]{\monot}{\type}
\newagentvar[\optimized]{\monorev}{\rev}

\newagentvar{\rcalloc}{y}
\newagentvar[\optimized]{\optrcalloc}{\rcalloc}
\newagentvar{\biddist}{G}
\newagentvarsm{\critbid}{\bid}{an}{constrained}
\newagentvar[\constrained]{\cbid}{B}
\newagentvar[\primedarg]{\wbid}{\bid}
\newagentvar{\gfunc}{\vartheta}

\newagentvarsm{\util}{u}{an}{}
\newagentvarsm{\vutil}{u}{interim}{forvalue}

\newcommand{\thebid}{b}
\newagentvar{\bid}{\thebid}
\newagentvarsm{\strat}{\thebid}{interim}{}

\newagentvar[\tagged{\text{SD}}]{\sdvirt}{\virt}
\newagentvar[\composed{\tagged{\text{SD}}}{\ironed}]{\sdivirt}{\virt}
\newagentvar[\tagged{\text{MD}}]{\mdvirt}{\virt}
\newagentvar[\composed{\tagged{\text{MD}}}{\ironed}]{\mdivirt}{\virt}

\newagentvar[\ironed]{\iprice}{\price}  %% what for?
\newagentvar[\ironed]{\ival}{\val}  %% what for?

\newagentvar{\ints}{{\cal I}}

\newagentvar{\wal}{w}

\newitemvar{\pos}{j}
\newitemvar{\weight}{w}
\newitemvar[\differentiated]{\mweight}{\weight}
\newitemvar{\udtype}{\type}
\newitemvar[\constrained]{\udcrittype}{\type}
\newitemvar{\udalloc}{\alloc}
\newitemvar{\udprice}{\price}
\newitemvar[\constrained]{\udcalloc}{\qalloc}
\newitemvar[\constrained]{\udcumcalloc}{\cumalloc}

\newagentvar{\mech}{{\cal M}}
\newagentvar[\skew{5}{\hat}]{\cmech}{{\cal M}}
\newagentvar{\alg}{{\cal A}}

\newagentvarsm{\rawprice}{\theprice}{expost}{forbid}
\newagentvarsm{\rawalloc}{\thealloc}{expost}{forbid}

\newcommand{\reals}{{\mathbb R}}

\newagentvar{\trans}{\sigma}

\newagentvar{\demandset}{S}

\newcommand{\opt}{^{\star}}
\newcommand{\primed}{^\dagger}

\newagentvar{distout}{w}
\newagentvar{idistout}{\bar{w}}

\newagentvar{\gap}{\delta}

\newagentvar[\tilde]{\ualloc}{\alloc}
\newagentvar[\tilde]{\uprice}{\price}

\newagentvarsm{\cumbal}{L}{an}{none}
\newagentvarsm{\maxbal}{d}{an}{ironed}
\newagentvarsm{\newbal}{d}{interim}{none}
\newagentvarsm{\rebal}{l}{interim}{none}
\newagentvarsm{\indnewbal}{d}{an}{random}

\newagentvarsm{\minalloc}{\iota}{none}{none}
\newagentvarsm{\rebalrate}{\eta}{an}{none}
\newcommand{\numrebalstages}{\random{\tau}}
\newcommand{\rebalstages}{\random{T}}

\newcommand{\insrate}{\rho}

\newcommand{\theinsalloc}{w}
\newcommand{\theinspay}{r}
\newagentvarsm{\insalloc}{\theinsalloc}{expost}{forval}
\newagentvarsm{\insvalloc}{\theinsalloc}{interim}{forval}
\newagentvarsm{\inspay}{\theinspay}{interim}{forval}
\newagentvarsm{\indinsalloc}{\theinsalloc}{an}{random}
\newagentvarsm{\indbelow}{\pi}{an}{random}
\newagentvarsm{\indalloc}{\thealloc}{an}{random}
\newagentvarsm{\indinspay}{\theinspay}{an}{random}
\newagentvarsm{\indval}{\theval}{an}{random}
\newagentvarsm{\indcumbal}{\cumbal}{an}{random}
\newagentvarsm{\indrebal}{\rebal}{an}{random}
\newagentvar{\inderror}{\omega}

\newagentvarsm{\empalloc}{x}{interim}{empirical}
\newagentvarsm{\empinsalloc}{w}{interim}{empirical}
\newagentvarsm{\empavg}{\mu}{none}{empirical}

\newagentvar[\underline]{\vmin}{\val}
\newagentvar[\bar]{\vmax}{\val}

\newcommand{\Ind}[1]{\mathbf{1}[#1]}

%% file: math.tex
\newcommand{\setsize}[1]{{\left|#1\right|}}

\DeclareMathOperator{\argmax}{argmax}

\newcommand{\super}[1]{^{(#1)}}

\newcommand{\given}{\,\mid\,}

\newcommand{\prob}[2][]{\text{\bf Pr}\ifthenelse{\not\equal{}{#1}}{_{#1}}{}\!\left[{\def\givenn{\middle|}#2}\right]}
\newcommand{\expect}[2][]{\text{\bf E}\ifthenelse{\not\equal{}{#1}}{_{#1}}{}\!\left[{\def\givenn{\middle|}#2}\right]}

\newcommand{\tparen}{\big}
\newcommand{\tprob}[2][]{\text{\bf Pr}\ifthenelse{\not\equal{}{#1}}{_{#1}}{}\tparen[{\def\given{\tparen|}#2}\tparen]}
\newcommand{\texpect}[2][]{\text{\bf E}\ifthenelse{\not\equal{}{#1}}{_{#1}}{}\tparen[{\def\given{\tparen|}#2}\tparen]}

\newcommand{\sprob}[2][]{\text{\bf Pr}\ifthenelse{\not\equal{}{#1}}{_{#1}}{}[#2]}
\newcommand{\sexpect}[2][]{\text{\bf E}\ifthenelse{\not\equal{}{#1}}{_{#1}}{}[#2]}

\newcommand{\dd}{{\mathrm d}}

%% file: abstract.tex
\begin{abstract}
This paper gives a theoretical model for design and analysis of
mechanisms for online marketplaces where a bidding dashboard enables
the bid-optimization of long-lived agents.  We assume that a good
allocation algorithm exists when given the true values of the agents
and we develop online winner-pays-bid and all-pay mechanisms that
implement the same outcome of the algorithm with the aid of a bidding
dashboard.  The bidding dashboards that we develop work in conjunction
with the mechanism to guarantee that bidding according to the
dashboard is strategically equivalent (with vanishing utility
difference) to bidding truthfully in the sequential truthful
implementation of the allocation algorithm.  Our dashboard mechanism
makes only a single call to the allocation algorithm in each stage.
\end{abstract}

%% file: intro.tex
\section{Introduction}
\label{s:intro}

%
% Problem space.
%
This paper formalizes a theoretical study of bidding dashboards for the
design of online markets.  In these markets short-lived users arrive
frequently and are matched with long-lived agents.  Rarely in practice
do market mechanisms, which allow the agents to bid to optimize
their matching with the users, have truthtelling as an equilibrium
strategy.\footnote{Such a matching is often the combination of
  a ranking of the agents by the market mechanism and the users'
  choice behavior.  Examples of the agents in such mechanisms include,
  sellers posting a buy-it-now price on eBay, hoteliers choosing a
  commission percentage on Booking.com, or advertisers bidding in
  Google's ad auction.  The first-two-examples are ``winner pays bid''
  for the agents; none of these examples have truthtelling as an
  equilibrium. Though the Google Ad auction is not winner-pays-bid,
  its bidding dashboard displays clicks versus cost-per-click which is
  equivalent to a winner-pays-bid mechanism.}  
The non-truthfulness of
mechanisms for these markets is frequently the result of practical
constraints; this paper is part of a small but growing literature on
the design of these non-truthful mechanisms.
By formally modeling bidding dashboards, which are common
in these markets, this paper gives simple and practical approach to
non-truthful mechanism design.

%
% two challenges: how should agents bid and what are good mechanisms
%
Two key challenges for the design of non-truthful mechanisms are (i)
assisting the agents to find good strategies in these mechanisms and
(ii) identifying mechanisms that have good outcomes when agents find
good strategies.  Bidding dashboards, which provide agents with market
information relevant for bid optimization and are common in online
marketplaces,\footnote{For example, Google provides a dashboard that
  forecasts click volume and cost per click as a function of
  advertisers' bids and Booking.com provides a visibility booster that
  forecasts percentage of increased clicks relative to base commission
  as a function of the commission percentage offered by the hotel for
  bookings.} can provide a solution to both issues.  This paper gives
a formal study of dashboard mechanisms and identifies dashboards and the
accompanying mechanism that lead agents to find good strategies and
in which good strategies result in good outcomes.

%
% main new idea: econometric inference
%
A novel technical feature of the mechanisms proposed by this paper is
the use of econometric inference to infer the preferences of the
agents in a way that permits the mechanism to then allocate nearly
optimally for the inferred preferences.  To illustrate this idea,
consider the Nash equilibrium of the winner-pays-bid mechanism that
allocates a single item to one of two agents with probability
proportional to their bids, i.e., an agent wins with probability equal
to her bid divided by the sum of bids, cf.\@ the proportional share
mechanism \citep[e.g.,][]{JT-04}.  Fixing the bid of agent~2, and
assuming agent 1's bid is in best response, the standard
characterization of equilibria in auctions allows the inversion of
agent 1's bid-optimization problem and identification of agent 1's
value.\footnote{The bid-allocation rule of agent~1 is
  $\balloci[1](\bid) = \bid / (\bid + \bidi[2])$.  Agent~1 with value
  $\vali[1]$ chooses $\bidi[1] = \argmax_{\bid} (\vali[1] - \bid)\,
  \balloci[1](\bid)$; this bid satisfies $\vali[1] = \bidi[1] +
  \balloci[1](\bidi[1]) / \balloci[1]'(\bidi[1]) = 2\bidi[1] +
  \bidi[1]^2/\bidi[2]$.  For example, if $\bidi[1] = 1$ and $\bidi[2]
  = 2$, we can infer that agent~1's value is $\vali[1] = 5/2$.}  The
same holds for agent~2.  Thus, in hindsight the principal knows both
agents' values. The difficulty, which this paper addresses, of using
this idea in a mechanism is the potential circular reasoning that
results from attempting to allocate efficiently given these inferred
values.

%
% utility model and mechanism model.
%
The paper considers general environments for single-dimensional
agents, i.e., with preferences given by a value for service and
utility given by value times service probability minus payment.  The
main goal of the paper is to convert an allocation algorithm, which
maps the values of the agents to the set of agents that are served,
into a mechanism that implements the same outcome but has
a winner-pays-bid or all-pay payment format.  In winner-pays-bid mechanisms the
agents report bids, the mechanism selects a set of agents to win, and
all winners pay their bids.  In all-pay mechanisms the agents report
bids, the mechanism selects a set of agents to win, and all agents pay
their bids.  For practical reasons, many markets are winner-pays-bid
or all-pay.\footnote{All-pay mechanisms correspond to markets were
  agents pay monthly for service and the service level depends on the
  total payment which is fixed in advance by the agent's bid.}

From a theoretical point of view, the focus on winner-pays-bid and
all-pay mechanisms is interesting because the prior literature has not
identified good winner-pays-bid mechanisms for general environments
(with the exception of \citealp{HT-16}, to be discussed with related
work). For example, the natural highest-bids-win winner-pays-bid
mechanism for the canonical environment of single-minded combinatorial
auctions does not exhibit equilibria that are close to optimal (even
for special cases where winner determination is computationally
tractable; see~\citealp{LB-10}, and \citealp{DK-15}).

%
% social choice rules
%
The motivating marketplaces for this work are ones were short-lived
users perform a search or collection of searches in order to identify
one or more long-lived agents with whom to transact.  The principal
facilitates the user search by, for example, ranking the agents by a
combination of perceived match quality for the user and the agents'
willingness to pay for the match, henceforth, the agents'
valuations. Strategically, the agents are in competition with each
other for matching with the user.  Viewing the user as a random draw
from a population of users with varying tastes results in a stochastic
allocation algorithm which maps the valuation profile of agents to the
probabilities that a user selects each agent.  This stochastic
allocation algorithm is typically monotone in each agent's value and
continuous \citep[e.g.,][]{AN-10}.

%
% the dashboard mechanism
%
A dashboard gives the agent information about market conditions that
enables bid optimization.  Without loss in our setting, the dashboard
is the agent's predicted bid allocation rule.  Our
dashboard mechanism assumes that the bids are best response to the
published dashboard, i.e., they ``follow the dashboard'', and uses
econometric inference to invert the profile of bids to get a profile
of values.  The dashboard mechanism for a given dashboard and
allocation algorithm is as follows:
\begin{enumerate}
\item[0.] Publish the dashboard bid allocation rule to each agent and solicit each agent's bid.
\item \label{step:invert-bids-with-signaled-allocation} Invert the bid
  of each agent, assuming it is in best response to the published
  dashboard, to obtain an inferred value for each agent.
\item \label{step:execute} Execute the allocation algorithm on the
  inferred values to determine the outcome; charge each winner her bid
  (winner-pays-bid) or all agents their bids (all-pay).
\end{enumerate}

Our main result is to identify dashboards where following the
dashboard is not an assumption but near optimal strategy for the agents.  For online
marketplaces, we define winner-pays-bid and all-pay dashboards for any
sequence of allocation algorithms where following the dashboard for a
sequence of values is approximately strategically equivalent to
bidding those values in the sequential truthful mechanism
corresponding to the sequence of allocation algorithms.  Thus, there
is an approximate correspondence between all equilibria.  The notion
of approximate strategic equivalence is that the allocations are
identical and the differences in payments vanish with the number of
stages.

Our dashboards can be implemented in the (blackbox) single-call model
for mechanism design.  In this model, an algorithm has been developed
that is monotonic in each agent's value and obtains a good outcome.
The mechanism's only access to the allocation algorithm is by live
execution, i.e., where the outcome of the algorithm is implemented.
\citet{BKS-10} introduced this model and showed that truthful
mechanisms can be single-call implemented.  We show that
winner-pays-bid and all-pay dashboard mechanisms can be single-call
implemented in online marketplaces.

The above approach to the design of online markets is simple and
practical.  Though we have specified the mechanism and framework for
winner-pays-bid and all-pay mechanisms and single-dimensional agents,
the approach naturally generalizes to multi-dimensional environments
and other kinds of mechanisms.  It requires only that, from bids that
are optimized for the dashboard, the preferences of the agents can be
inferred.  Moreover, payment formats can be mixed-and-matched across
different agents within the same mechanism.

%% file: related_work.tex
\paragraph{Related Work.}

%
% comparison to Hartline Taggart
%

This paper is on non-revelation mechanism design as it applies to
online markets.  The goals of this paper -- in providing foundations
for non-revelation mechanism design -- are closest to those of
\citet{HT-16,HT-19}.  Both papers study iterated environments and the design
of mechanisms with welfare that is arbitrarily close to the optimal
welfare under general constraints.  (\citeauthor{HT-16} additionally
consider the objective of maximizing revenue.)  The equilibrium
concept of \citeauthor{HT-19} is Bayes-Nash equilibrium and it is
assumed that agents are short-lived.  These assumptions are not
appropriate for the design of online markets considered in this paper,
where agents are long-lived and possibly have persistent values.

There have been extensive studies of the price of anarchy of ad
auctions, e.g., on the Google search
engine \citep{LT-10,CKKK-11,ST-13}.  A conclusion of these studies is
that in standard models of equilibrium the industry-standard
generalized second price auction can have welfare that is a constant
factor worse than the optimal welfare.  In these auctions, dashboards
can also be used to infer agent values and the mapping from bids to
values can be used in place of the ``quality score'' that is already in
use in these mechanisms.  Dashboard-based mechanisms like the ones
we construct could obtain more efficient outcomes in these markets.

The approach of the main positive result of the paper is to link
payments across successive stages in the sequential mechanism.
Linking otherwise independent decisions has proven to be useful in
overcoming difficulties in many areas of mechanism
design.  \citet{JS-07} show that linking decisions in i.i.d.\
sequential bilateral exchange can bypass the \citet{MS-83}
impossibility result. \citet{GBI-17} shows that good mechanisms
without money can be designed in a sequential environment by endowing
agents with an artificial currency and running the sequential truthful
mechanism with money. The endowment of artificial currency effectively
links decisions across stages. \citet{arm-99} shows that selling
independent items in a bundle can raise higher revenue than selling
separately, cf.\ \citet{BILW-14}.  The linking decisions idea is also
in play in the previously discussed work on Bayes-Nash mechanism
design by \citet{HT-19}.

The results of this paper can also be viewed as giving a blackbox
reduction from truthful mechanism design to algorithm design.  For
single-dimensinal agents, a Bayes-Nash reduction from mechanism design
to non-monotone algorithm design was given by \citet{HL-10,HL-15}.
Multi-dimensional reductions for addressing non-monotonicity were
studied by \citet{HKM-15}; \citet{BH-11}; and \citet{DHKN-17}.
Blackbox reductions of revenue maximization to welfare maximization
were studied by \citet{CDW-13a,CDW-13b}.  Single-call reductions,
which reduce truthful mechanism design to algorithm design with a
single blackbox call to the monotone algorithm were introduced
by \citet{BKS-10} and further studied by \citet{WS-15}.  Our work
extends this last area by giving single-call reductions to monotone
algorithms for winner-pays-bid and all-pay mechanisms.

%% file: overview.tex
\paragraph{Organization.} 
The rest of this paper is organized as follows.
\Cref{s:prelim} formalizes the multi-agent and sequential environment for mechanism design. 
\Cref{s:single-agent} reviews single-agent auction theory.  In particular, it shows that implementing winner-pays-bid and all-pay mechanisms for a single agent is straightforward.  
\Cref{s:dashboards} gives the basic design of a dashboard
mechanism.
\Cref{s:inferred-values-dashboard} gives a
dashboard and proves that, in static environments, following the
dashboard converges to Nash equilibrium and this equilibrium
implements the algorithm's desired allocation.
\Cref{s:rebalancing} 
defines our dynamic solution concept for dashboard mechanisms, namely
approximate strategic equivalence with the sequential truthful
mechanism.  In this section we describe a payment rebalancing
dashboard that guarantees this approximate strategic equivalence.
\Cref{s:minimal-rebalancing} shows that there are natural
all-pay and winner-pays-bid dashboards for which following the dashboard with a static
value results in low outstanding balance, even in a dynamic
environment and without explicit rebalancing.
Finally, \Cref{s:blackbox}
generalizes the dashboards to environments where the mechanism only
has access to the allocation algorithm by single blackbox call.

%% file: prelim.tex
\section{The Sequential Multi-agent Model}
\label{s:prelim}

This paper considers general environments for single-dimensional
linear agents.  There are $n$ agents $\{1,\ldots,n\}$.  A profile of
$n$ agent values is denoted $\vals = (\vali[1],\ldots,\vali[n])$.  A
multi-agent allocation algorithm is $\epvallocs : \reals^n \to
[0,1]^n$.  We assume that any constraints of the environment are
satisfied by the algorithm and will not otherwise be notating these
constraints.  A multi-agent non-truthful mechanism is denoted
$(\epballocs,\epbpays)$ with bid allocation rule $\epballocs :
\reals^n \to [0,1]^n$ and $\epbpays : \reals^n \to \reals^n$.  The
agents have linear utility, for example, the utility of agent $i$ on
bids $\bids$ in $(\epballocs,\epbpays)$ is $\vali\,\epballoci(\bids) -
\epbpayi(\bids)$.

We consider agents interacting sequentially in a mechanism over many
stages.  For stage $s$ in $\{1,\ldots,t\}$, the environment is given
by stage values $\vals \super s$ and a stage allocation algorithm
$\epvallocs \super s$.  The designer chooses a stage mechanism
$(\epballocs \super s, \epbpays \super s)$ and the agents choose stage
bids $\bids \super s$.  The designer may choose the stage mechanism
based on realized information and outcomes from previous stages. 

We consider dynamic environments were the values and allocation
algorithm can be distinct across stages and static environments where
the values and the allocation algorithm are the same in every stage.
In static environments both the designer may change the mechanism and
the agents may change their bids from stage to stage.

An allocation algorithm takes several forms which will be notated
distinctly.  The ex post allocation algorithm is $\epvallocs :
\reals^n \to [0,1]^n$.  The ex post allocation for agent $i$ is
$\epvalloci : \reals^n \to [0,1]$.  Denote the profile with agent
$i$'s value replaced with $\zee$ is $(\zee,\valsmi) =
(\vali[1],\ldots,\vali[i-1],\zee,\vali[i+1],\ldots,\vali[n])$.  Fixing
the other agent values $\valsmi$, the projection of the ex post
allocation algorithm is single-agent allocation rule $\valloci(\zee) =
\epvalloci(\zee, \valsmi)$.  An allocation which could be output from
$\epvallocs$ is $\allocs$; agent $i$'s allocation is $\alloci$.
Similar notation is adopted for other quantities that pertain to
agents.  For example the single-agent strategy function that maps agent
$i$'s value to a bid is denoted by $\strati$ while a bid is denoted by
$\bidi$.

%% file: single-agent.tex
\section{Single-Agent Implementation}
\label{s:single-agent}

The multi-agent dashboard mechanisms of this paper will be based on a
straightforward construction of winner-pays-bid and all-pay mechanisms
for a single agent.  The agent has a single-dimensional value $\val$
and linear utility $\util = \val\,\alloc - \price$ for allocation
probability $\alloc$ and expected payment $\price$.  Given a stochastic
allocation algorithm $\valloc : \reals \to [0,1]$, which specifies a
desired mapping from the agent's value to an allocation, we wish to
implement $\valloc$ when the agent behaves strategically in a
mechanism with a winner-pays-bid or all-pay payment format.

A single-agent mechanism $(\balloc,\bpay)$ maps a bid $\bid$ to an
allocation probability $\balloc(\bid)$ and payment $\bpay(\bid)$.
Winner-pays-bid and all-pay mechanisms specify a bid allocation rule
$\balloc$ with bid payment rule $\bpay$, respectively, defined as
\begin{align}
\label{eq:bid-payment-format}
\bpay(\bid) &= \bid\,\balloc(\bid), &
\bpay(\bid) &= \bid.
\end{align}

Strategic behavior of an agent in such a mechanism is governed by the
$n=1$ agent special case of the characterization of Bayes-Nash
equilibrium of \citet{mye-81}.  See \Cref{f:myerson} for an
illustration of the geometry of the payment identity.

\begin{figure}[t]
\begin{center}
\begin{tikzpicture}

\draw (1.25,1.2) node {payment};

\coordinate (xofv) at (1.3,.5);

\begin{scope}
\clip (0,0) rectangle (xofv);

\fill [orange] (0,0) -- plot [smooth] coordinates {(.3,0) (.5,.2) (.7,.22) (1,.25) (1.2,.33) (xofv) (1.5,.6) (1.6,.62) (1.7,.64) (1.8,.8) (1.9, .85) (1.98,.98) (2,1) } -- (0,1) -- cycle;
\end{scope}

\draw [thick] (0,0) -- plot [smooth] coordinates {(.3,0) (.5,.2) (.7,.22) (1,.25) (1.2,.33) (xofv) (1.5,.6) (1.6,.62) (1.7,.64) (1.8,.8) (1.9, .85) (1.98,.98) (2,1) } -- (2.5,1);

% ticks

\xlabel{1.3}{$\val$}
\ylabel{.5}{$\valloc(\val)$}
\xlabel{0}{$0$}
\ylabel{0}{$0$}
\ylabel{1}{$1$}

% axis
\draw (0,1) -- (0,0) -- (2.5,0);

\end{tikzpicture}
\begin{tikzpicture}

\draw (1.25,1.2) node {surplus};

\coordinate (xofv) at (1.3,.5);

\fill [orange] (0,0) rectangle (xofv);

\draw [thick] (0,0) -- plot [smooth] coordinates {(.3,0) (.5,.2) (.7,.22) (1,.25) (1.2,.33) (xofv) (1.5,.6) (1.6,.62) (1.7,.64) (1.8,.8) (1.9, .85) (1.98,.98) (2,1) } -- (2.5,1);

% ticks

\xlabel{1.3}{$\val$}
\ylabel{.5}{$\valloc(\val)$}
\xlabel{0}{$0$}
\ylabel{0}{$0$}
\ylabel{1}{$1$}

% axis
\draw (0,1) -- (0,0) -- (2.5,0);

\end{tikzpicture}
\begin{tikzpicture}

\draw (1.25,1.2) node {utility};

\coordinate (xofv) at (1.3,.5);

\begin{scope}
\clip (0,0) rectangle (xofv);

\fill [cyan] plot [smooth] coordinates {(.3,0) (.5,.2) (.7,.22) (1,.25) (1.2,.33) (xofv) (1.5,.6) (1.6,.62) (1.7,.64) (1.8,.8) (1.9, .85) (1.98,.98) (2,1) } -- (2.5,1) -- (2.5,0) -- cycle;
\end{scope}

\draw [thick] (0,0) -- plot [smooth] coordinates {(.3,0) (.5,.2) (.7,.22) (1,.25) (1.2,.33) (xofv) (1.5,.6) (1.6,.62) (1.7,.64) (1.8,.8) (1.9, .85) (1.98,.98) (2,1) } -- (2.5,1);

% ticks

\xlabel{1.3}{$\val$}
\ylabel{.5}{$\valloc(\val)$}
\xlabel{0}{$0$}
\ylabel{0}{$0$}
\ylabel{1}{$1$}

% axis
\draw (0,1) -- (0,0) -- (2.5,0);

\end{tikzpicture}
\caption{\label{f:myerson} The terms in the payment identity are
  depicted with $\price(0) = 0$.  The payment of the agent is equal to
  surplus $\val\,\valloc(\val)$ minus the utility of the agent
  $\int_0^{\val} \valloc(\zee)\,\dd \zee$.}
\end{center}
\end{figure}
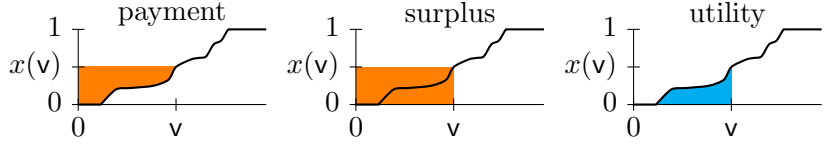

\begin{theorem}[A special case of \citealp{mye-81}]
\label{t:myerson}
Bidding strategy $\strat: \reals \to \reals$ is a best response for
mechanism $(\balloc,\bpay)$ if and only if its induced
allocation and payment rules $\valloc(\cdot) = \balloc(\strat(\cdot))$
and $\vpay(\cdot) = \balloc(\strat(\cdot))$ satisfy:
\begin{enumerate}[(a)]
\item monotonicity: $\valloc(\cdot)$ is monotonically non-decreasing;
\item payment identity: 
\begin{align}
\label{eq:payment-identity}
\bpay(\val) 
   &= \val\,\valloc(\val) - \int_0^{\val} \valloc(\zee) \, \dd \zee 
   + \vpay(0);
\end{align}
\end{enumerate}
and, for any value of the agent, any bid not in the range of the
bidding strategy is dominated by a bid in its range.
\end{theorem}

A simple consequence of \Cref{t:myerson} is that with a monotone
allocation algorithm $\valloc$ the mechanism $(\valloc,\vpay)$, with
$\vpay$ given by the payment identity, is {\em truthful}, i.e., bidding $\strat(\val) = \val$ is an optimal
strategy for the agent.  We will refer to $(\valloc,\vpay)$ with
$\vpay$ satisfying the payment identity as the truthful
implementation of $\valloc$.  Often in the literature the payment that
the agent makes with value $\val=0$ is $\vpay(0) = 0$.

The payment identity of \Cref{t:myerson} also allows the
winner-pays-bid and all-pay mechanisms that implement a monotone
allocation algorithm $\valloc$ to be derived.  In
equations~\eqref{eq:single-agent-strat}
and~\eqref{eq:single-agent-bid-allocation-rule} these derivations are
given with winner-pays-bid on the left and all-pay on the right.
Specifically, the payment identity \eqref{eq:payment-identity} and the
formulae for payments from the payment format
\eqref{eq:bid-payment-format} give two equations for payments from
which the bid strategy can be derived:
\begin{align}
\label{eq:single-agent-strat}
\strat(\val)&= \val - \frac{1}{\valloc(\val)} \int_0^{\val}
\valloc(\zee)\,\dd \zee + \frac{\vpay(0)}{\valloc(\val)}, & \strat(\val) &= \val\, \valloc(\val)-\int_0^{\val}\valloc(\zee)\,\dd \zee + \vpay(0).\\ 
\intertext{Bid strategy $\strat$ is invertible when it is strictly increasing (to be further discussed below). From
  this inversion $\strat^{-1}$, the bid allocation rules that
  implement the allocation rule $\valloc$ are:}
\label{eq:single-agent-bid-allocation-rule}
\balloc(\bid) &= \valloc(\strat^{-1}(\bid)), & \balloc(\bid) &= \valloc(\strat^{-1}(\bid)).
\end{align}
Note that the winner-pays-bid and all-pay strategies $\strat$ from
equation~\eqref{eq:single-agent-strat} as applied in
equation~\eqref{eq:single-agent-bid-allocation-rule} are distinct.

The above discussion assumes that we are given $\valloc$ and which to
construct a winner-pays-bid or all-pay mechanism $(\balloc,\bpay)$.
In this case, inferring values from bids can be accomplished with the
inverse bidding strategy $\strat^{-1}$ from
equation~\eqref{eq:single-agent-strat}.  If instead we were given
$(\balloc,\bpay)$ we could invert values from bids by evaluating the
agent's first order condition as follows.  With value $\val$ the
agent's optimal bid $\bid$ in single-agent mechanism $(\balloc,\bpay)$
satisfies $\val = \bpay'(\bid) / \balloc'(\bid)$.  Evaluating this
formula with the winner-pays-bid and all-pay payment formats of
equation~\eqref{eq:bid-payment-format}, respectively, we have:
\begin{align}
\label{eq:single-agent-foc}
\val &= \bid + \balloc(\bid)/\balloc'(\bid), & \val &=
1/\balloc'(\bid). 
\end{align}

Thus, we see from any desired allocation rule $\valloc$ there is a
corresponding payment rule $\vpay$, bid-allocation rule $\balloc$, and
bid strategy $\strat$.  Unless explicitly stated otherwise, the
payment rule, bid-allocation rule, and bid strategy correspond to
transfer $\vpay(0) = 0$.  In \Cref{s:dashboards}, we will define
dashboards based on single-agent allocation rule $\ealloc$ with
corresponding payment rule $\epay$, bid-allocation rule $\dalloc$, and
bid strategy $\dstrat$.  These notational conventions will be adopted
throughout the paper.  A key comparison will be of the desired allocation rule with $\valloc$, $\vpay$, $\balloc$, and $\strat$ against the dashboard with $\ealloc$, $\epay$, $\dalloc$, and $\dstrat$.

In the above derivation of winner-pays-bid and all-pay implementations
of the allocation rule $\valloc$, it was assumed that the bid strategy
$\strat$ is strictly increasing.  When the allocation rule $\valloc$
is strictly increasing and with $\vpay(0) \leq 0$, $\strat$ is
strictly increasing for both all-pay and winner-pays-bid
implementations.  When the transfer $\vpay(0) > 0$, is bid strategy
$\strat$ is strictly increasing for all-pay implementations but not
increasing for winner-pays-bid implementations.

\begin{lemma}
  \label{l:p0}
  For strictly increasing allocation rule $\valloc$, all-pay bid
  functions $\strat$ are strictly increasing.  For strictly increasing
  allocation rule $\valloc$ and non-positive-transfer $\vpay(0) \leq
  0$, winner-pays-bid bid functions $\strat$ are strictly increasing.
\end{lemma}
\begin{proof}
  The all-pay result follows from equality of the bid function
  $\strat$ and payment rule $\vpay$ and the strict monotonicity of the
  latter under the assumption that the allocation rule $\valloc$ is
  strictly monotone.  See \Cref{f:myerson}.

  The winner-pays-bid result follows from differentiating the bid
  function and observing that the derivative is positive if and only
  if the surplus is at least the payment, i.e., $\val\,\valloc(\val) >
  \vpay(\val)$ for all $\val \geq 0$.  Specifically, by the geometry
  of the payment identity (\Cref{f:myerson}), $\vpay(\val) \leq 0$
  implies $\val\,\valloc(\val) > \vpay(\val)$; while $\vpay(\val) > 0$
  implies $\val\,\valloc(\val) < \vpay(\val)$ at $\val=0$.
  \end{proof}

The dashboard mechanisms designed in the remainder of the paper
effectively convert the mechanism design problem of implementing a
multi-agent allocation algorithm into a collection of problems where
the above single-agent derivations apply.

%% file: dashboards.tex
\section{Dashboard Mechanisms}

\label{s:dashboards}

This section defines a family of dashboard mechanisms that give a
practical approach to bidding and optimization in non-truthful
Internet markets.  The principal publishes a bidding dashboard which
informs the agents of the market conditions; the agents use this
dashboard to optimize their bids.  We consider dashboards that give
each agent estimates of the outcome of the mechanism for any possible
bid of the agent.  These estimates can be constructed, for example,
from historical bid behavior of the agents in the mechanism.

\begin{definition}\label{d:dashboard}
For a single-agent incentive compatible mechanism $(\ealloc,\epay)$,
the single-agent {\em dashboard} is $\dalloc : \reals \to [0,1]$ as
given by equation~\eqref{eq:single-agent-bid-allocation-rule} and the
bid-strategy for the dashboard is $\dstrat : \reals \to \reals$ as
given by equation~\eqref{eq:single-agent-strat}.\footnote{Thus, the
    notation for dashboard $\ealloc$ of payment rule $\epay$,
    bid-allocation rule $\dalloc$, and bid-strategy $\dstrat$
    parallels the notation for allocation rule $\valloc$ with $\vpay$,
    $\balloc$, and $\strat$.}  A multi-agent dashboard is a profile of
  single-agent bid-allocation rules $\dallocs =
  (\dalloci[1],\ldots,\dalloci[n])$ where $\dalloci(\bid)$ is the
  forecast probability that agent $i$ wins with bid~$\bid$.
\end{definition}

We consider the implementation of an allocation algorithm with a
monotone allocation rule
$\epvallocs : \reals^n \to [0,1]^n$ by winner-pays-bid and all-pay
dashboard mechanisms.  The winner-pays-bid and all-pay format
dashboards are identical except with respect to equations
\eqref{eq:single-agent-strat}-\eqref{eq:single-agent-foc} which
respectively define the value-to-bid optimal bidding function, the
translation from (value) allocation rules to bid allocation rules, and
the bid-to-value inversion for a bid-allocation rule.

We consider two design questions: (a) what dashboard should the
principal publish and (b) what mechanism should the principal run.
The goal is to pick a dashboard for which following the
dashboard is a good strategy and if agents follow the
dashboard then the allocation algorithm is implemented.  In fact, if
agents follow the dashboard, then the mechanism to implement the
allocation algorithm is straightforward.

\begin{definition}
\label{d:dashboard-mechanism}
The dashboard mechanism $\epballocs : \reals^n \to
[0,1]^n$ for dashboard $\dallocs$ and allocation algorithm $\epvallocs$
is:
\begin{enumerate}
\item[0.] Solicit bids $\bids$ for dashboard $\dallocs$.
\item \label{s:dashboard-inference} Infer values $\evals$ from bids
  via equation~\eqref{eq:single-agent-foc}.
% as $\evali = \bidi +
%  \dalloci(\bidi) / \dalloci'(\bidi)$ for each $i$.
\item Output allocation $\epballocs(\bids) = \epvallocs(\evals)$ (with prices $\epbpays(\bids)$ according to the payment format).
\end{enumerate}
\end{definition}

Notice that in the above definition the bid allocation rule
$\epballocs$ is a mapping $\reals^n \to [0,1]^n$ from bid profiles to
allocation profiles while the dashboard $\dallocs$ is a profile of
single-agent bid allocation rules $(\reals \to [0,1])^n$ each of which
maps a bid to an allocation probability.  In other words, outcomes
according to the dashboard are calculated independently across the
agents while outcomes according to the mechanism depend on the reports
of all agents together.  (The notation distinguishes these kinds of
allocation rules by using distinct typeface.)  The following
proposition follows simply from the correctness of
equation~\eqref{eq:single-agent-foc}.

\begin{prop} 
In the dashboard mechanism (\Cref{d:dashboard-mechanism}) for any
given strictly monotone and continuously differentiable dashboard and any
allocation algorithm, if the agents follow the dashboard then the
allocation algorithm is implemented.
\end{prop}

%% We consider the dashboard mechanism for a given allocation algorithm
%% in an iterated environment with agents with persistent values, i.e.,
%% that are unchanging over each iteration.  In practice these iterations
%% could correspond to a new user arriving in the market to be matched
%% with the agents or an aggregation of such arrivals over, e.g., one
%% billing cycle of the agents.  In this environment, the principal is
%% able to tune the dashboard $\dallocs$ employed with the dashboard
%% mechanism of \Cref{d:dashboard-mechanism}.  Notice that changing the
%% dashboard changes the definition of the mechanism; thus the game
%% played in each stage is not the same game.

%% file: inferred-values-dashboard.tex
\section{Static Environments: Inferred Values Dashboards}
\label{s:static}
\label{s:inferred-values-dashboard}

Our analysis of dashboards in this section is restricted to static
settings where both the agents' values and the allocation algorithm do
not change from stage to stage.  How should the principal construct
the dashboard to ensure the sequential dashboard mechanism, if all
agents follow the dashboard, converges so that following the dashboard
is a Nash equilibrium?  For static environments, we will say that
following the dashboard {\em converges to Nash equilibrium} if,
assuming other agents follow the dashboard, an agent's best response
in the stage game converges to following the dashboard.

Our approach is motivated by the solution concept of fictitious play.
In fictitious play agents best respond to the empirical distribution
of the actions in previous rounds.  Fictitious play assumes that the
agents know the actions in past rounds.  The principal could publish
these actions; however, a better approach is to just publish as a
dashboard the aggregate bid allocation rules that result.  Our
dashboard follows this approach except with respect to estimated
values rather than actions.  Following such a dashboard is,
in a sense, an improvement on fictitious play.

\begin{definition}
\label{d:inferred-values-dashboard}
For stage $t+1 \geq 2$, the {\em inferred values dashboard} is the
profile of single-agent bid allocation rules $\dallocs \super {t+1}$
defined as follows:
\begin{enumerate}
\item The inferred valuation profile for stage $s \in \{1,\ldots,t\}$ is $\evals \super s$.

\item The profile of single-agent allocation rules for stage $s \in \{1,\ldots,t\}$ is $\vallocs \super s$ with $\valloci \super s (\zee) = \epvalloci(\zee,\evalsmi \super s)$ for each agent $i$.
\item The empirical profile of single-agent allocation rules at stage $t+1$ is $\eallocs \super {t+1}$ with $\ealloci \super {t+1} (\zee) = \frac{1}{t} \sum_{s\leq t} \valloci \super s (\zee)$ for each agent $i$.
\item 
The inferred values dashboard is the profile of single-agent bid
allocation rules $\dallocs \super {t+1}$ that correspond to profile
$\eallocs \super {t+1}$ via
equation~\eqref{eq:single-agent-bid-allocation-rule}.
\end{enumerate}
\end{definition}

\begin{definition} 
\label{d:last-stage}
The {\em $k$-lookback inferred values dashboard} is the variant of the
inferred values dashboard that averages over the last $\min(k,t)$
stages.  The {\em last-stage inferred values dashboard} chooses $k=1$.
\end{definition}

The inferred values dashboard does not specify a dashboard for stage
1.  For stage 1, any strictly increasing dashboard will suffice.  When
the agents' values and the allocation rule are static, following the
dashboard converges to Nash equilibrium.  

%% Suppose the agents follow the
%% dashboard in stage 1, then the inferred values for stage 1 are
%% correct, and the allocation algorithm is implemented correctly, albeit
%% with incorrect payments. The dashboard published for stage 2 is the
%% correct dashboard for the allocation algorithm at the correctly
%% inferred values from stage 1.  Following the dashboard in stage 2 is a
%% Nash equilibrium: Assuming the other agents follow the dashboard (and
%% the same values are inferred as in stage 1), an agent faces the same
%% allocation rule as is represented in the stage 2 dashboard and
%% following the dashboard is a best response.

\begin{theorem}
\label{t:Nash-convergence}
In static environments with fixed values $\vals$ and continuous and
strictly monotone stage allocation algorithm $\epvallocs$, in the
sequential dashboard mechanism with the $k$-lookback inferred values
dashboard, if agents follow the dashboard in stages $\max(t-k,1)$
through $t>1$ then following the dashboard in stage $t+1$ is a Nash
equilibrium.
\end{theorem}

\begin{proof}
The dashboards up to and including round $t$ are the average of the
bid allocation rule for a continuous and strictly increasing allocation algorithm; thus they are continuous and strictly increasing.  By
equation~\eqref{eq:single-agent-foc} the inferred values of agents
that follow the dashboard are the true values.  These values are the
same in each stage; thus the profile of single-agent allocation rules
in each stage is the one that corresponds to the allocation algorithm on true valuation profile.  The average of these profiles of
allocation rule is the profile itself (they are all the same).  The
dashboard is the corresponding profile of single-agent bid allocation
rules.

Consider agent $i$ and assume that other agents are following the
dashboard in stage $t+1$.  Thus, the estimated profile of other agent
values is $\evalsmi\super{t+1} = \valsmi$.  The allocation rule in
value space faced by agent $i$ is $\valloci \super {t+1} (\zee) =
\epvalloci(\zee,\valsmi)$ which is equal to the allocation rule in all the
previous stages.  As the dashboard suggests bidding optimally
according to the allocation rule of the previous stages, bidding
according to the dashboard is agent $i$'s best response.  Thus, following
the dashboard is a Nash equilibrium in stage $t+1$.
\end{proof}

This dashboard mechanism is simple and practical and can implement any
strictly monotone and continuously differentiable allocation
algorithm.  For example, welfare maximization with a convex regularizer
gives such a allocation algorithm.  For the paradigmatic problem of
single-minded combinatorial auction environments, our dashboard
mechanism with the regularized welfare-maximization allocation algorithm gives outcomes that can be arbitrarily close to optimal (for
the appropriate regularizer).

A critical issue with the guarantee of \Cref{t:Nash-convergence} is
that it is delicate to the Nash assumption.  If other agents do not
follow the dashboard, then following the dashboard is not necessarily
a good strategy for the agent.  The remainder of the paper will
resolve this issue by giving stronger analyses and stronger
dashboards.  Specifically, we will give a dashboard for which the
sequential dashboard mechanism is strategically equivalent to the
sequential truthful mechanism.

%% file: rebalancing.tex
\section{Dynamic Environments: Payment Rebalancing Dashboards}
\label{s:rebalancing}

In this section we develop a dashboard that will satisfy the strong
guarantee that the $\epsilon$ equilibria of the sequential dashboard
mechanism are approximately the same as the $\epsilon$ equilibria from
running the truthful mechanism that implements the allocation
algorithm in each stage.  Moreover, the follow-the-dashboard strategy
corresponds to the truthtelling strategy and, thus, is an
$\epsilon$-equilibrium.\footnote{Note that the sequential truthful
  mechanism will generally have many other equilibria as well.}

\begin{definition}
\label{d:strategic-equiv}
Two sequential mechanisms are {\em $\epsilon$ approximately
  strategically equivalent} if for all sequences of bids profiles
in one mechanism there is a corresponding sequence of bid profiles
in the other mechanism such that the absolute per-stage-average
difference in utilities of any agent in the two mechanisms is at most
$\epsilon$, and vice versa.
\end{definition}

\begin{definition}
\label{d:sequential}
For a sequence of monotone stage allocation algorithms $\epvallocs \super 1,
\ldots, \epvallocs \super t$, the sequential truthful mechanism is
given by the stage mechanism $(\epvallocs \super s,\epvpays \super s)$
for stage $s \in \{1,\ldots,t\}$ where $\epvpayi \super s(\vals)$ is
defined from the payment identity applied to
$\epvalloci(\cdot,\valsmi)$ for each agent $i$.
\end{definition}

\begin{definition}
\label{d:incentive-consistent}
A sequential dashboard mechanism is {\em $\epsilon$ incentive
  consistent} if it is $\epsilon$ strategically equivalent to the
sequential truthful mechanism and the follow-the-dashboard strategy
corresponds to the truthtelling strategy.
\end{definition}

%% By the following
%% theorem, it will be sufficient to develop a dashboard that is
%% incentive consistent for all strategies.

%% \begin{theorem}
%% \label{t:equivalence}
%% A dashboard mechanism that is $\epsilon$-incentive inconsistent for
%% all strategies is strategically equivalent up to utility $\epsilon$ to
%% a sequential truthful mechanism with the same sequence of allocation
%% algorithms $\vallocs \super 1,\ldots,\vallocs \super t$.
%% \end{theorem}
%% \begin{proof}
%% Strategy profiles in the dashboard mechanism are in one-to-one
%% correspondence with sequences of inferred valuation profiles $\evals
%% \super 1,\ldots,\evals \super t$.  If each agent reports their
%% sequence of inferred values in the truthful mechanism, incentive
%% consistency implies that the utility that each agent receives is
%% within $\epsilon$ of their utility in the dashboard mechanism.  Thus,
%% all approximate equilibria in the sequential truthful mechanism are
%% approximate equilibria of the dashboard mechanism and vice versa.
%% \end{proof}

We now develop a dashboard for which the sequential dashboard
mechanism is $\epsilon$ incentive consistent with $\epsilon$ that
vanishes with the number of stages.  The high-level approach of this
dashboard is simple: In any stage where the actual payment and the
truthful payment are different, add the payment residual to a balance
and adjust future dashboards to either collect additional payment or to
discount the required payment so that the outstanding balance is resolved
over a few subsequent stages.  The difference in utilities of an agent
in such a dashboard mechanism and the truthful mechanism with the same
allocation algorithm is bounded by the per-stage payment residual and
the number of stages it takes to resolve it.  When these quantities
are both constants, the average per-stage difference between the
dashboard mechanism's outcome and the truthful mechanism's outcome
vanishes with the number of stages.

The existence of this dashboard shows that in a repeated scenario
there is essentially no difference between winner-pays-bid, all-pay,
and truthful payment formats.  In some sense, linking payments between
stages and adjusting the mapping from bids to values allows a
market designer to choose a payment format that is appropriate for
the application.

Three final comments before presenting the mechanism: First, balancing
payments in all-pay mechanisms is much easier than balancing payments
in winner-pays-bids mechanisms.  The reason is simple, in all-pay
mechanisms the payment is deterministic and, thus, any additional
payment requested is paid exactly.  On the other hand, payments
collected in winner-pays-bid mechanisms depend on the probability of
allocation.  If this probability of allocation is fluctuating then
collected payments can over- or under-shoot a target.  We give an
approach below that resolves this issue.  Second, in the static
setting where the agent values and the allocation algorithm are
unchanging, in the follow-the-dashboard equilibrium of the inferred
values dashboard mechanism (\Cref{d:inferred-values-dashboard}) the only non-trivial payment residual is in
the first stage, once this balance is resolved, the accrual of
subsequent balance is off the equilibrium path.  Thus, in steady state
the rebalancing required by the dashboard is trivial.  Third,
dashboards with payment rebalancing allow dynamically changing
environments, e.g., agent values and the allocation algorithm.  For
agents that follow the dashboard, the rebalancing mechanism only kicks
in when the agent's value or the environment changes.

The main definition and proposition that motivate the approach are as
follows.

\begin{definition}
Define the {\em outstanding balance} of an agent in a $t$-stage
sequential dashboard mechanism for inferred values $\evals\super 1,
\ldots, \evals \super t$ as the magnitude of the total expected
difference in payments between the sequential dashboard mechanism and
the sequential truthful mechanism when agents report values
$\evals\super 1, \ldots, \evals \super t$.
\end{definition}

\begin{prop}
\label{p:eps-ic}
A $t$-stage dashboard mechanism with worst-case (over valuation
profiles) outstanding balance upper bounded by $\epsilon t$ is
$\epsilon$ incentive consistent.
\end{prop}

\begin{proof}
Fix $t$ valuation profiles $\evals \super 1, \ldots,\evals \super t$
and consider the allocation and payments from the sequential dashboard
mechanism where agents follow the dashboard for these valuation
profiles and the sequential truthful mechanism where agents report
these valuation profiles.  The allocations obtained by these two
mechanisms are identical as both call the stage $s$ allocation
algorithm $\epvallocs \super s$ on valuation profile $\evals \super
s$.  By the assumption of the proposition, the outstanding balance of
any agent, i.e., the total difference in payments of the mechanisms,
is upper bounded by $\epsilon t$.  For any true valuation profiles,
$\vals \super 1, \ldots, \vals \super t$, and because the utilities of
the agent are linear in payments; the per-stage-average difference in
utility between the two mechanisms is at most $\epsilon$.  Thus, the
definition of $\epsilon$ incentive consistency is satisfied.
\end{proof}

\subsection{Rebalancing Dashboards}

In the subsequent developments we will focus on a single agent.  We
are given a dashboard with alloction rule $\ealloc$ (and corresponding
payment rule $\epay$, bid allocation rule $\dalloc$, and bid strategy
$\dstrat$).  The rebalancing method will construct a new dashboard
that corresponds to $\ealloc\primed$, $\epay\primed$,
$\dalloc\primed$, and $\dstrat\primed$.  The general approach of the
construction is that, if the allocation algorithm and dashboard
correspond with $\valloc = \ealloc$ then the dashboard is correct and
the correct payments are made, and otherwise, the dashboard is incorrect
and the incorrect payments will need to be resolved by future
dashboards.

The main idea in the proposed rebalancing approach is that adding a
constant to the agent's expected payment, i.e., setting transfer
$\epay(0) \neq 0$ in the payment identity, does not affect incentives.
Unfortunately, there are two potential issues with this approach.
Crucially, by \Cref{l:p0}, winner-pay-bid bid functions are not
invertible for $\epay(0) > 0$.  Thus, such a dashboard mechanism
cannot be implemented as described thus far.  Moreover, with $\epay(0)
> 0$ the stage mechanism is not individually rational.  This latter
concern is not as serious as the former concern as an agent's long run
utility will be non-negative from participating even when the agent's
utility in a single stage could be negative.  Our approach will be to
view these dashboards as a reference for subsequently developed
dashboards which satisfy $\epay(0) = 0$ and both admit invertible bid
functions and satisfy individual rationality.

\begin{definition}\label{d:rebalancing-dashboard}
For a single-agent truthful mechanism $(\ealloc,\epay)$,
the single-agent {\em reference rebalancing dashboard} for rebalancing rate
$\rebalrate \in (0,1]$ and outstanding balance $\cumbal$ is 
$\dalloc\primed$ that corresponds to payment rule $\epay\primed$
defined as $\epay\primed(\val) = \epay(\val) +
\cumbal\,\rebalrate$, i.e., with $\epay\primed(0) = \cumbal\,\rebalrate$.
\end{definition}

From \Cref{d:rebalancing-dashboard} and equation~\eqref{eq:single-agent-strat}, the bidding strategy can be
calculated for winner-pays-bid and all-pay dashboards, respectively
\begin{align}
\label{eq:rebalancing}
\dstrat\primed(\val) &= \dstrat(\val) + \cumbal\, \rebalrate / \ealloc(\val),&
\dstrat\primed(\val) &= \dstrat(\val) + \cumbal\,\rebalrate.
\end{align}
The bid-allocation rule of the dashboard $\dalloc\primed$ is then
defined via the allocation rule $\ealloc$ and the inverse of the
strategy $\dstrat\primed$ via
equation~\eqref{eq:single-agent-bid-allocation-rule}.  The bid of an
agent should be viewed as shown in equation~\eqref{eq:rebalancing} as
two terms. The first term is for the original dashboard $\ealloc$ and the
second term is for resolving the outstanding balance.  After each
stage the balance is adjusted to account for how much of the
outstanding balance was resolved and by any new payment residual
resulting from misestimation of the dashboard $\ealloc$ for the
realized allocation rule $\valloc$.  While our analysis will keep track
of when payment residuals are generated and how long it takes to
resolve them, the balance tracking need only consider the difference
between what was paid and what should have been paid for the realized
allocation rule.

\begin{definition} 
\label{d:payment-residual}
For allocation rule $\valloc$ with bid-strategy $\strat$, dashboard
allocation rule $\ealloc$, and inferred value $\eval$, the {\em
  payment residual} for the reference rebalancing dashboard $\balloc\primed$ with realized allocation rule $\valloc$ is
\begin{align}
\label{eq:newbal}
\newbal(\eval) &= [\strat(\eval) - \dstrat(\eval)]\, \valloc(\eval), &
\newbal(\eval) &= \strat(\eval) - \dstrat(\eval)\\
\intertext{in winner-pays-bid and all-pay formats, respectively.  The
  {\em balance resolved} is}
\label{eq:rebal}
\rebal(\eval) &= [\cumbal\, \rebalrate / {\ealloc(\eval)}]\, \valloc(\eval), &
\rebal(\eval) &= \cumbal\, \rebalrate,
\end{align}
respectively.  The total change to the balance is $\newbal(\eval) -
\rebal(\eval)$.
\end{definition}

Note that in the above definition for winner-pays-bid mechanisms, if
$\valloc = \ealloc$ then $\strat = \dstrat$ and $\newbal(\eval) = 0$.
For all-pay mechanisms with $\valloc = \ealloc$, then $\strat =
\dstrat$ and $\newbal(\eval) = 0$.  The perspective to have is that in
steady-state it should be that $\valloc = \ealloc$ and there should be
no payment residual but, when agents arrive, depart, or have value
changes, then there can be non-trivial payment residual which will be
rebalanced by the dashboard.  More generally the following lemma
bounds the magnitude of the payment residual.

\begin{lemma}
\label{l:maximum-imbalance}
In a stage in which the agent's inferred value is $\eval$, the
magnitude of payment residual $|\newbal(\eval)|$ is at most $\eval$.
\end{lemma}
\begin{proof}
From equation \eqref{eq:newbal}, for all-pay mechanisms the
payment residual is the difference in bids in an individually rational
mechanism.  As bids for an agent with value $\eval$ are between $0$
and $\eval$, the magnitude of their difference is at most $\eval$.
For winner-pays-bid $\strat(\eval) - \dstrat(\eval)$ is
again a difference of bids and multiplying this difference by
$\valloc(\eval) \in [0,1]$ preserves the bound of $\eval$ on its
magnitude.
\end{proof}

The analysis in the sections below makes no assumption about whether
the mechanism is in steady-state.  Both the allocation algorithm and
the agents' values may change from stage to stage.  The main idea of
the analysis is to consider the payment residual $\newbal \super s$ in
stage $s$ and calculate how much of it can remain after stage $t$ when
there is a guaranteed fraction of it is rebalanced as part of $\rebal
\super {s'}$ for subsequent stages $s' \in \{s+1,\ldots,t\}$.  We give
a worst case analysis that pessimisticly assumes that the payment
residual in each stage is the same sign and equal to its maximum
value, i.e., the value of the agent.  These bounds apply to the
rebalancing dashboard constructed from any original dashboard that may
not have any relation to the realized allocation rule.

\subsection{Outstanding Balance of the All-pay Rebalancing Dashboard} 

\begin{definition}
From stage $s$ with outstanding balance $\cumbal \super s$, agent bid
$\bid \super s$, inferred value $\eval \super s$, realized allocation
rule $\valloc \super s$, and corresponding all-pay bid strategy
$\strat \super s$; the stage $s+1$ outstanding balance in the all-pay
rebalancing dashboard is:
$$
\cumbal \super {s+1} = \cumbal \super s + \strat \super s (\eval \super s) - \bid \super s.
$$
\end{definition}

The analysis of the rebalancing dashboard for all-pay mechanisms is
trivial.  To parallel the subsequent presentation for winner-pays-bid
mechanisms below, we write the formal lemma and theorem.

\begin{lemma}
\label{l:all-pay-rebalancing}
In the all-pay dashboard with rebalancing rate $\rebalrate \in (0,1]$,
  in any stage the balance resolved is $\rebal(\eval) =
  \cumbal\,\rebalrate$; for $\rebalrate = 1$ the full outstanding
  balance is resolved.
\end{lemma}
\begin{proof}
See equation~\eqref{eq:rebal}.
\end{proof}

\begin{theorem}
\label{t:cumulative-imbalance-all-pay}
In dynamic environments, the sequential dashboard mechanism with an
all-pay rebalancing dashboard with rebalancing rate
$\rebalrate = 1$ and per-stage estimated value at most $\vmax$;
the outstanding balance at stage $t$ is at most $\vmax$
and, consequently, over $t$ stages the dashboard mechanism is
$\vmax/t$ incentive consistent.
\end{theorem}

\begin{proof}
\Cref{l:maximum-imbalance} upper bounds the payment residual of each
stage by $\vmax$ (which is added to the outstanding balance).
\Cref{l:all-pay-rebalancing} upper bounds the outstanding balance
prior to stage $s$ that is resolved in stage $s$.  With $\rebalrate =
1$ the outstanding balance at stage $t$ is only the payment residual
from stage $t$.  Thus, the outstanding balance after stage $t$ is
upper bounded by $\vmax$.  Applying \Cref{p:eps-ic} the sequential
dashboard mechanism is $\vmax/t$ incentive consistent.
\end{proof}

Note that \Cref{t:cumulative-imbalance-all-pay} can be improved in
static environments where both inferred stage values $\eval \super s$
and the single-agent allocation rules induced from the values of the
other agents $\valloc \super s$ are static, i.e., $\eval \super s =
\eval$ and $\valloc \super s = \valloc$.  In these cases the inferred
values dashboard (\Cref{d:inferred-values-dashboard}) is $\ealloc
\super s = \valloc$ in all stages but the first and only this first
stage accrues payment residual.  Thus, the total balance at stage $t
\geq 2$ with $\rebalrate = 1$ is $\cumbal \super t = 0$.

\subsection{Outstanding Balance of the Winner-pays-bid Rebalancing Dashboard} 

Defining good rebalancing dashboards is more challenging for
winner-pays-bid mechanisms as the balance resolved depends on the
dashboard at the estimated value $\ealloc(\eval)$ which we do not
require to be constant across rounds.  We therefore chose a
conservative rebalancing rate $\rebalrate < 1$ to avoid overshooting
the target of zero balance.  Specifically, we will set $\rebalrate$
less than the minimum allocation probability of the dashboard; thus
$\rebalrate/\ealloc(\eval) \leq 1$ in equation~\eqref{eq:rebal}.  As
the rebalancing dashboard can fix any incorrect dashboard, it may be
desirable to distort the dashboard at the low end to guarantee that
that the minimum allocation probability of the dashboard is not too
small.

Since we are considering winner-pays-bid mechanisms we will only
adjust an agent's balance when the agent is allocated.  To make
explicit the actual values versus their expectations we will adopt
notation $\indcumbal$ and $\indalloc$ for the actual balance and
actual allocation with $\cumbal = \expect{\indcumbal}$ and
$\valloc(\eval) = \expect{\indalloc}$.  We rewrite, per the discussion
above, the payment residual and balanced resolved from equations \eqref{eq:newbal} and~\eqref{eq:rebal} for inferred value
$\eval$ as:
\begin{align}
\label{eq:newindbal}
\indnewbal &= [\strat(\eval) - \dstrat(\eval)]\, \indalloc, &
\indrebal &= [\cumbal\, \rebalrate / {\ealloc(\eval)}]\, \indalloc
\end{align}
where $\expect{\indnewbal} = \newbal(\eval)$ and $\expect{\indrebal} =
\rebal(\eval)$.  Note that $\strat$ is the
winner-pays-bid bid strategy corresponding to $\valloc$ while $\dstrat$
is the winner-pays-bid bid strategy for dashboard $\ealloc$.

\begin{definition}
From stage $s$ with outstanding balance $\cumbal \super s$, agent bid
$\bid \super s$, inferred value $\eval \super s$, realized allocation rule $\valloc
\super s$, corresponding winner-pays-bid bid strategy $\strat \super s$, and realized allocation $\indalloc \super s$; the
stage $s+1$ outstanding balance in the all-pay rebalancing dashboard 
is:
$$
\indcumbal \super {s+1} = \indcumbal \super s + [\strat \super s(\eval \super s)  - \bid \super s]\, \indalloc \super s.
$$
\end{definition}

Denote by $\rebalstages = \{s : \indalloc \super s \neq 0\}$ the set
of stages where the agent is allocated and by $\numrebalstages =
\setsize{\rebalstages}$ the number of such stages.  These are the
stages where $\indnewbal \super s$ and $\indrebal \super s$ can be
non-zero.  We consider the amount of the payment residual $\indnewbal
\super s$ from stage $s \in \rebalstages$ that remains at final stage
$t$ when, at each subsequent stage in $\rebalstages$, a fraction of that
payment residual is resolved.

\begin{lemma}
\label{l:remaining-imbalance}
At any rebalancing stage $s \in \rebalstages$ and under any agent
strategy, the winner-pays-bid rebalancing dashboard with rebalancing rate $\rebalrate
\in (0,1)$ for dashboard allocation rule $\ealloc \super
s$ with allocation probability supported on $[\rebalrate,1]$ resolves
balance $\indrebal$ between $\indcumbal\,\rebalrate$ and
$\indcumbal$.
\end{lemma}

\begin{proof} 
This result follows from the definition of $\indrebal$ and the fact
that $\rebalrate \leq \rebalrate / \ealloc(\eval) \leq 1$ for all $\eval$ by the
assumption that $1 \geq \ealloc(\eval) \geq \rebalrate$.
\end{proof}

%% \begin{lemma}
%% \label{l:remaining-imbalance}
%% For allocation rules with allocation probabilities supported on
%% $[\minalloc,1]$, the winner-pays-bid dashboard with $\rebalrate \in
%% (0,\minalloc)$, the imbalance rebalanced $\rebal$ in any stage and under any
%% agent strategy is between $\cumbal\,\rebalrate^2$ and $\cumbal$.
%% \end{lemma}

%% \begin{proof}
%% The expected rebalanced payment of the agent is
%% $\valloc(\eval)\,\cumbal\, \rebalrate / \ealloc(\eval)$.  The maximum value
%% of the rebalanced payment is $\cumbal\, \rebalrate / \minalloc \leq
%% \cumbal$ which occurs when $\valloc(\eval) = 1$ and $\ealloc(\eval) =
%% \minalloc$.  The minimum value of the rebalanced payment is
%% $\minalloc\, \cumbal\, \rebalrate \geq \cumbal\, \rebalrate^2$ which occurs
%% when $\valloc(\eval) = 1$ and $\ealloc(\eval) = \minalloc$.
%% \end{proof}

The following theorem about the winner-pays-bid dashboard rebalancing
mechanism will upper bound the outstanding balance at any time.  With
these quantities taken as constants relative to the number of stages
$t$, the imbalance per stage is vanishing with $t$.  It is useful to
contrast the assumptions and bounds of the analogous result for
all-pay dashboards (\Cref{t:cumulative-imbalance-all-pay}) with
\Cref{t:cumulative-imbalance}.

\begin{theorem}
\label{t:cumulative-imbalance}
For any monotonic stage allocation rules $\valloc \super 1,\ldots,
\valloc \super t$ and any monotonic dashboard rules $\ealloc \super 1,
\ldots, \ealloc \super t$ with probabilities supported on
$[\rebalrate,1]$, the winner-pays-bid payment rebalancing dashboard
with rebalancing rate $\rebalrate$ and per-stage payment residual at
most $\maxbal$ has outstanding balance at stage $t$ of at most
$\maxbal / \rebalrate$.
\end{theorem}

\begin{proof}
If we start a stage $s$ with outstanding balance $\indcumbal \super
s$, \Cref{l:remaining-imbalance} implies that at least a $\rebalrate$
(and at most 1) fraction of it is resolved.  The balance remaining
is at most $(1-\rebalrate)\,\indcumbal \super s$ (and at least zero).

Denote the stages that the agent is allocated, indexed in decreasing
order, by $\rebalstages = \{s_0,\ldots,s_{\numrebalstages-1}\}$.  The
payment residual $\indnewbal \super {s_k}$ from stage $s_k$ that
remains in the final stage $t$ is at most
$(1-\rebalrate)^k\,\indnewbal \super {s_k}$. (This bound is worst case,
specifically, we do not track the possibility that some of the balance
might cancel with new payment residual of the opposite sign.)  Summing over each stage $s_k \in \rebalstages$, we can
bound the outstanding balance at stage $t$ by:
\begin{align*}
\cumbal \super t &\leq \maxbal \sum\nolimits_{k=0}^{\numrebalstages-1} (1-\rebalrate)^k\\
                 &\leq \maxbal \sum\nolimits_{k=0}^{\infty} (1-\rebalrate)^k\\
                 &= \maxbal/\rebalrate. \qedhere
\end{align*}
\end{proof}

With an upper bound on the per-stage estimated value of $\vmax$,
\Cref{t:cumulative-imbalance} can be combined with
\Cref{l:maximum-imbalance}, which bounds $\maxbal \leq \vmax$, to obtain the following corollary.  The
sequential dashboard mechanism with a rebalancing dashboard is
$\epsilon$ incentive consistent with $\epsilon \leq
\vmax/(\rebalrate\, t)$ in $t$ rounds (see
\Cref{d:incentive-consistent} and \Cref{p:eps-ic}).  With
$\vmax/\rebalrate$ held as constants the incentive inconsistency
vanishes with $t$.

\begin{corollary}
\label{c:cumulative-imbalance}
In dynamic environments, the sequential dashboard mechanism with a
winner-pays-bid rebalancing dashboard with rebalancing rate
$\rebalrate \in (0,1)$, dashboard allocation probabilities supported
on $[\rebalrate,1]$, and per-stage estimated value at most $\vmax$;
the outstanding balance at stage $t$ is at most $\vmax / \rebalrate$
and, consequently, over $t$ stages the dashboard mechanism is
$\vmax/(\rebalrate\, t)$ incentive consistent.
\end{corollary}

Note that \Cref{t:cumulative-imbalance} can be improved in static
environments where both inferred stage values $\eval \super s$ and the
single-agent allocation rules induced from the values of the other
agents $\valloc \super s$ are static, i.e., $\eval \super s = \eval$
and $\valloc \super s = \valloc$.  In these cases the inferred values
dashboard (\Cref{d:inferred-values-dashboard}) is $\ealloc \super s =
\valloc$ in all stages but the first and only this stage accrues
non-trivial payment residual.  In this case, the total balance at
stage $t$ is $\cumbal \super t \leq
\eval\,(1-\rebalrate)^{\numrebalstages-1}$, i.e., exponentially small
in the number of stages that the agent is allocated.

\subsection{Individual Rationality and Non-negative Transfers}

The rebalancing dashboards presented and analyzed above set a transfer
$\epay(0) \neq 0$ so that payment residuals of early stages are
resolved in subsequent stages.  {\em Individual rationality} requires
that $\epay(0) \leq 0$ and must be satisfied for winner-pays-bid
bidding strategies to be invertible.  {\em Non-negative transfers}
requires that $\epay(0) \geq 0$.  Typical mechanisms satisfy both
properties, i.e., $\epay(0)
= 0$.  

In this section we describe winner-pays-bid dashboards with transfer
$\epay(0) = 0$; the main idea is for the dashboard outcome to look to
high-valued agents like $\epay(0) \neq 0$ but, for low-valued agents
where such an outcome is not feasible with $\epay(0) = 0$, to have
payment as close as possible to the $\epay(0) = 0$ outcome.  Notice that the
highest individually rational payment for an agent with value $\val$
is $\val$; the lowest non-negative payment for an agent is 0.  Thus,
we seek to identify allocation rules with approximately these payments
for low-valued agents to combine with the original allocation rule for
high-valued agents. 

The following observation which follows from the
derivation~\eqref{eq:single-agent-strat} of the winner-pays-bid bid
strategy of $\dstrat(\zee) = \epay(\zee) / \ealloc(z)$ motivates the
developments of this section.  Specifically, scaling the dashboard
allocation $\ealloc$ by some factor $\alpha > 0$ similarly scales the
payments $\epay$ by the same factor; the bid function is their ratio
which is constant with respect to this scaling.
\begin{fact}
  \label{f:scale-indep}
  Winner-pays-bid bid strategies are independent of the scale of the
  allocation rule, i.e., the bid strategies $\dstrat$ and
  $\dstrat\primed$ corresponding to $\ealloc$ and $\ealloc \primed$,
  defined as $\ealloc\primed(\zee) = \alpha\,\ealloc(\zee)$ for any
  $\alpha > 0$, are identical.
\end{fact}

We define a family of allocation rules which admit simple linear bid
strategies.  An allocation rule in the family can be selected to
require either a very high bid or a very low bid, depending on whether
we need to rebalance a positive or negative balance.

\begin{definition}
  The $\gamma \in (0,1)$ {\em linear-bid allocation rule} is $\ealloc^{\gamma}(\zee) = \zee^{\nicefrac{\gamma}{1-\gamma}}$; 
  the linear-bid allocation rule that meets point
  $(\val,\alloc)$ is $\ealloc^{\gamma,\val,\alloc}(\zee) = [\nicefrac{\alloc}{\ealloc^{\gamma}(\val)}]\,\ealloc^{\gamma}(\zee)$.
\end{definition}

\begin{lemma}
  The bid strategy for the $\gamma$ linear-bid allocation
  rule $\ealloc^{\gamma,\val,\alloc}$ is $\dstrat^{\gamma,\val,\alloc}(\zee) = \dstrat^{\gamma}(\zee) = \gamma\,\zee$.
\end{lemma}

\begin{proof}
  The proof follows by calculation of the bid function according to
  equation~\eqref{eq:single-agent-strat}.  By \Cref{f:scale-indep}, it
  suffices to calculate the bid function for $\ealloc^{\gamma}$.
\end{proof}

The payment rebalancing dashboard is constructed from an original
dashboard by replacing its allocation to low values with an
appropriately scaled linear-bid allocation rule.

\begin{definition}
\label{d:rebalancing-dashboard-no-p0}
  The payment rebalancing dashboard $\ealloc\primed$ for outstanding balance $\cumbal$ and original dashboard $\ealloc$ is constructed as follows.\footnote{Notes: The desired $\val\primed$ exists as $\dstrat$ is continuously increasing but finite.  In the case where $\cumbal = 0$, the definition sets $\val\primed = 0$ and $\ealloc\primed = \ealloc$.}    
  \begin{enumerate}
  \item For non-negative balance $\cumbal \geq 0$, set $\gamma \to 1$,
    i.e., just below 1; and set $\val\primed$ to solve $\val\primed -
    \dstrat(\val\primed) = \cumbal$.
  \item For negative balance $\cumbal < 0$, set $\gamma \to 0$, i.e.,
    just above 0; and set $\val\primed$ to solve $-\dstrat(\val\primed) =
    \cumbal$.
  \item Set dashboard allocation rule $\ealloc\primed$ as:
  \begin{align*}
    \ealloc\primed(\zee) &=
    \begin{cases}
      \ealloc^{\gamma,\val\primed,\ealloc(\val\primed)}(\zee) & \text{for $\zee \in [0,\val\primed]$}\\
      \ealloc(\zee) &\text{for $\zee \in [\val\primed,\infty)$.}
    \end{cases}
  \end{align*}
  \end{enumerate}
 \end{definition}

Rebalancing dashboards constructed via
\Cref{d:rebalancing-dashboard-no-p0} satisfy the following properties.
\begin{enumerate}
\item \label{case:low}
  Low-valued agents with value $\val \leq \val\primed$ bid nothing for
  $\cumbal < 0$ or bid their full value for $\cumbal > 0$.
  
\item \label{case:high}
  High-valued agents with value $\val \geq \val\primed$ bid according
  to the rebalancing bid function $\dstrat\primed$ of
  \Cref{d:rebalancing-dashboard}.
\end{enumerate}

Our analysis of the outstanding balance for high-valued agents will be
identical to the previous discussion of the reference rebalancing
dashboard of \Cref{d:rebalancing-dashboard}.  Our analysis of the
outstanding balance for low-valued agents is based on the following
simple lemma which shows that either the balance is improved or the
magnitude of the outstanding balance is at most the inferred value
$\eval$, i.e., the same as the bound on the payment residual of
\Cref{l:maximum-imbalance}.

\begin{lemma}
  With outstanding balance $\cumbal$, dashboard $\ealloc$, allocation
  rule $\valloc$, and inferred value $\eval \leq \val\primed$ from \Cref{d:rebalancing-dashboard-no-p0}; the
  magnitude of the subsequent outstanding balance of the rebalancing dashboard
  mechanism is either reduced or at
  most $\eval$.
\end{lemma}

\begin{proof}
  We prove the case of positive outstanding balance $\cumbal > 0$ (the
  negative balance case is analogous).  We aim to show that the
  outstanding balance after this stage is between $-\eval$ and
  $\cumbal$, i.e., the balance is reduced and if it becomes negative
  then its magnitude is at most $\eval$.

  For $\gamma \to 1$ the bid according to dashboard $\ealloc\primed$
  is $\dstrat\primed(\eval) \to \eval$.  This is the most the agent
  could bid with value $\eval$ so $\strat(\eval) - \dstrat\primed(\eval)$ is non-positive.  Thus, the balance is reduced.

  It remains to show that if the outstanding balance becomes negative, i.e.,
  we overshoot zero balance, that its magnitude is at most $\eval$.
  The largest value of $\setsize{\strat(\eval) - \dstrat\primed(\eval)}$ is
  $\eval$ which occurs when $\strat(\eval) = 0$.  Thus, the remaining outstanding
  balance is at least $\cumbal - \eval$ which, since we started
  with a positive balance $\cumbal > 0$, is at least $-\eval$.  
\end{proof}

To analyze the outstanding balance at stage $t$ to prove an analog to
\Cref{t:cumulative-imbalance}, we break the %winning 
stages where the agent is allocated into three
cases:
\begin{enumerate}
\item \label{case:low-same-sign}
  The agent has a low inferred value and after the stage the outstanding
  balance does not change sign.
  
\item \label{case:low-change-sign}
  The agent has a low inferred value and after the stage the
  outstanding balance changes sign.

\item \label{case:high}
  The agent has a high inferred value.
\end{enumerate}

Our analysis will be simple.  In the first case, the magnitude of the outstanding
balance does not increase.  These stages can be ignored.  Combining
the stages where the other two cases hold, our previous analysis of
\Cref{t:cumulative-imbalance} can be applied.  To apply
\Cref{t:cumulative-imbalance} we need the balance resolved to be
within $[\cumbal\,\rebalrate,\cumbal]$ and the magnitude of the payment residual, i.e.,
new balance that is generated, to be at most $\eval$.  For stages that
satisfy the second case we view the payments as exactly balancing
$\cumbal$ and then generating a new payment residual with magnitude at most
$\eval$.  For stages of the third case, the assumptions of the theorem
are satisfied by the equivalence of the dashboard with the reference
rebalancing dashboard (\Cref{d:rebalancing-dashboard}) for the
inferred value $\eval$.  We conclude with the following theorem.

\begin{theorem}
\label{t:cumulative-imbalance-no-p0}
For any monotonic stage allocation rules $\valloc \super 1,\ldots,
\valloc \super t$, any monotonic dashboard rules $\ealloc \super 1,
\ldots, \ealloc \super t$ with probabilities supported on
$[\rebalrate,1]$, and any sequence of inferred values $\eval \super
1,\ldots,\eval \super t$ bounded by $\vmax$; the winner-pays-bid
rebalancing dashboard has outstanding balance at stage $t$ of at most
$\vmax / \rebalrate$; the rebalancing dashboard mechanism is $\vmax /
(\rebalrate t)$ incentive consistent.
\end{theorem}

\subsection{Practical Considerations}

In practical implementations of the rebalancing dashboard, it may be
undesirable for the dashboard to fluctuate significantly from stage to
stage as positive and negative balances are resolved.  One approach to
provide a dashboard that varies less across successive stages is to
allow the magnitude outstanding balance to be a small positive
multiple of $\eval$ and only resolve the outstanding balance when this
allowed magnitude is exceeded.  A detailed treatment of this approach
is given in the next section.

%% file: minimal-rebalancing.tex
\section{Static Values: Natural Rebalancing}
\label{s:minimal-rebalancing}

In this section we assume as in \Cref{s:static} and
\Cref{s:rebalancing} that the functional form of the allocation
algorithm is available.  Like \Cref{s:static} the value of the agent
of our analysis will be assumed to be static; like
\Cref{s:rebalancing} the stage allocation rule will be assumed to be
dynamic.

The rebalancing dashboard of \Cref{s:rebalancing} allowed the
conversion of the dashboard corresponding to any strictly monotone
allocation rule into one where the outstanding balance, i.e., the
difference between payments in the sequential dashboard mechanism and
sequential truthful mechanism, in any dynamic environment is bounded.
On the other hand, following the dashboard for the inferred values
dashboard of \Cref{s:static} converges in two stages to Nash and the
subsequent payment residual is zero.  In the static-value
dynamic-allocation-rule environment, we identify dashboard with
naturally low outstanding balance for an agent with static value that
follows the dashboard.  Thus, on the follow-the-dashboard equilibrium
path, explicit rebalancing is only necessary when an agent's value
changes.

When combining the naturally rebalancing dashboards of this section with
the rebalancing dashboards of \Cref{s:rebalancing}, it is important
not to attempt to resolve the payment residual immediately.  For
example, permitting a small outstanding balance of, e.g., $2\vmax$, to
persist will allow the naturally rebalancing dashboard to balance
itself.  Of course, if the agent's value changes, the payment residual
can accumulate and will need to be resolved by explicit rebalancing.

\subsection{All-pay Dashboards}

The main result of this section is the observation that the all-pay
last-stage inferred-values dashboard (\Cref{d:last-stage}) does not
require any rebalancing for an agent with a static value, even when
the allocation algorithm or other agent values can be dynamic.

\begin{lemma} 
\label{l:last-stage-outstanding-balance}
For an agent with static per-stage value $\val$, the outstanding
balance of the all-pay last-stage inferred-values dashboard mechanism
at stage $t$ is at most $\val$.
\end{lemma} 

\begin{proof}
We will consider grouping the allocation of stage $s$ with the payment
in stage $s+1$ (with the allocation of stage $t$ paired with the
payment of stage $1$) and we will refer to this grouping as ``outcome
$s$.''  If the agent follows the dashboard, outcomes corresponding to $s
\in \{1,\ldots,t-1\}$ are the truthful outcomes for the allocation
rule of stage $s$ and value $\val$.  Specifically, the bid in stage
$s+1$ is equal to the payment according to the dashboard in stage
$s+1$ which is equal to the allocation rule of stage $s$.  Thus, the
payment is stage $s+1$ is the correct payment for the allocation rule
in stage $s$ and static value $\val$.  The imbalance of payments from
outcome $t$ (the allocation from stage $t$ and payment from stage 1),
by \Cref{l:maximum-imbalance} in \Cref{s:rebalancing}, is at most
$\val$.
\end{proof}

As is evident by the proof of \Cref{l:last-stage-outstanding-balance},
the payment residual in each stage can be large, but because of the
relationship between the dashboard and allocation algorithm in
successive stages the overall outstanding balance remains small.

\subsection{Winner-pays-bid Dashboards}

A similar approach admits a winner-pays-bid dashboard that does not
require any rebalancing for an agent with a static value, even when
the allocation algorithm or other agent values can be dynamic.

As a general principle, winner-pays-bid dashboard mechanisms tend work
better when stages where an agent is not allocated are ignored.
While, ignoring these stages when constructing the dashboard may make
estimates of allocation rules less responsive, considering these
stages in dynamic environments leaves open the possibility that the
allocation rules for rounds where the agent is not allocated are
systematically distinct from allocation rules in rounds where the
ageht is allocated.  This distinction can result in systematic
differences between the dashboard and the allocation rule in stages
where the agent is allocated and can result in large differences
between the total payments in the sequential dashboard mechanism and
the sequential truthful mechanism.  The following dashboard is
constructed based on only allocation rules where the agent wins.

\begin{definition}
  The {\em last-winning-stage inferred-values dashboard} in stage $t$
  corresponds to the allocation rule from the latest stage $s < t$
  where the agent was allocated by the mechanism.
\end{definition}

\begin{lemma} 
\label{l:last-winning-stage-outstanding-balance}
For an agent with static per-stage value $\val$, the outstanding
balance of the winner-pays-bid last-winning-stage inferred-values dashboard mechanism at stage $t$ is at most $\val$.
\end{lemma} 

\begin{proof}
Denote the stages that the agent is allocated, indexed in increasing
order, by $\rebalstages = \{s_1,\ldots,s_{\numrebalstages}\}$.  For
convenience denote by $s_0 = 0$ and denote by $\valloc \super {s_0} =
\ealloc \super {s_1}$ the initial dashboard.  (This initial dashboard
can be any arbitrary monotonic allocation rule.)

The total payments made by the agent are $\sum_{k=1}^{\numrebalstages}
\dstrat \super {s_k}(\val)$.  Here $\dstrat \super {s_k}$ is the bid
strategy used in stage $s_k$ with the dashboard for $\ealloc \super {s_k}$.
The bid strategy corresponding to the stage allocation rule $\valloc \super {s_k}$ is
$\strat \super {s_k}$; thus, the desired total payments are
$\sum_{k=1}^{\numrebalstages} \strat \super {s_k} (\val)$.

For the last-winning-stage inferred-values dashboard $\ealloc \super
{s_k} = \valloc \super {s_{k-1}}$; thus, $\dstrat \super {s_k} =
\strat \super {s_{k-1}}$ Thus, the difference in the total payments
made and total payments desired is
\begin{align*}
  \sum\nolimits_{k=1}^{\numrebalstages}
  \dstrat \super {s_k}(\val)
  -   \sum\nolimits_{k=1}^{\numrebalstages}
  \strat \super {s_k}(\val)
  = \strat \super {s_0} (\val) - \strat \super {s_{\numrebalstages}}(\val).
\end{align*}
The magnitude of the difference between any two bids of an agent with value $\val$ is at most
$\val$.  The lemma follows.
\end{proof}

%% file: blackbox.tex
\section{Single-call Dashboards}
\label{s:blackbox}

In this section we generalize the construction of dashboards and their
analyses to the practically realistic case that the principal has only
single-call access to the allocation algorithm.  For example, given a
profile of values, the principal can draw a single sample from the
distribution of the algorithm's outcomes.  Such a dashboard is
necessary for online marketplaces where the allocations of the
algorithm are endogenous to the behavior of one side of the market.
E.g., in ad auctions where the mechanism is designed for the
advertisers and the users show up and make decisions that realize the
stochasticity of the allocation.  The single-call perspective of this
section allows viewing the allocation algorithm as a map from the
values of the agents to how agents are prioritized in the marketplace
where stochastic outcomes are obtained.

The single-call model of mechanism design is one where the designer
has an algorithm and aims to design a mechanism that implements the
allocation of the algorithm in equilibrium.  The designer, however,
can only call the algorithm once.  The single-call model is
significant in that the standard methods for calculating payments in
mechanisms, e.g., in the Vickrey-Clarke-Groves or unbiased payment
mechanism \citep{AT-01}, require $n+1$ blackbox calls to the
algorithm.  \citet{BKS-10} showed that truthful single-call mechanisms
exist.  These mechanisms randomly perturb the allocation algorithm and
use these perturbations to compute the correct payments for the
perturbed allocation algorithm.

There are two challenges to single-call implementation of dashboards.
First, for the valuation profile input into the algorithm, only the
realized allocation in $\indallocs \in \{0,1\}^n$ is observed.  The allocation
probabilities are not observed.  Second, the principal does not have
counterfactual access to the allocation algorithm that we have
previously used to determine reasonable dashboards.  Recall, a
dashboard is a prediction of the bid allocation rule faced by a agent.
This prediction requires knowledge, e.g., for agent $i$ of
$\epvalloci(z,\valsmi)$ for all $z$.

Both of these challenges are solved by instrumenting the allocation
algorithm.  Specifically, we instrument the allocation algorithm with
uniform exploration, i.e., with probably $\insrate$ independently for
each agent we enter a uniform random value rather than the agent's
value.  This instrumentation degrades the quality of the allocation
algorithm by $\insrate$.
%it also, as required previously in the payment
%rebalancing dashboard, enforces a lower bound on the minimum
%allocation probability of $\minalloc = \insrate$.  
First, this uniform instrumentation can be viewed as a randomized
controlled experiment for estimating the counterfactual allocation
rule as is necessary for giving the agent a dashboard that estimates
her bid-allocation rule.  Second, the uniform exploration enables
implicitly calculating unbiased truthful payments for the
realized mechanism \citep[cf.][]{BKS-10}.  The difference between
these desired payments and the actual payments, i.e., the agent's bid
if she wins, can then be added to the outstanding balance in the payment
rebalancing dashboard of \Cref{s:rebalancing}.  

\begin{definition}
\label{d:instrumented-scf}
The instrumented allocation algorithm $\insallocs$ for allocation
algorithm $\epvallocs$, instrumentation rate $\insrate \in (0,1)$,
valuation range $[0,\vmax]$, and input valuation profile $\vals$ is:
\begin{enumerate}
\item For each agent $i$, sample 
\begin{align*}
\indvali & \sim \begin{cases}
  \vali            & \text{with probability $1-\insrate$,}\\
 U[0,\vmax]  & \text{with probability $\insrate$.}
\end{cases}
\end{align*}
\item Define $\insallocs(\vals) = \expect[\indvals]{\epvallocs(\indvals)}$ and sample $\indinsallocs \sim \insallocs(\vals)$, i.e., run $\epvallocs$ on $\indvals$.
\item For each agent $i$, set instrumentation variables
\begin{align*}
\indalloci &= \indinsalloci\, \Ind{\indvali = \vali}, 
& \indbelowi &= \indinsalloci \, \tfrac{\vmax}{\vali} \, \Ind{\indvali < \vali}, 
&\indinspayi &= \vali \, \big( \indalloci - \tfrac{1-\insrate}{\insrate}\,\indbelowi\big).
\end{align*}
\end{enumerate}
The truthful payment functions for $\insallocs$ are
denoted by $\inspays$, via the payment identity~\eqref{eq:payment-identity}.
\end{definition}

\begin{theorem}\label{thm:unbiased} 
The instrumentation payment variables $\indinspays$ are unbiased
estimators for the incentive compatible payments $\inspays(\vals)$
for the instrumented allocation algorithm $\insallocs$, i.e.,
$\expect{\indinspays} = \inspays(\vals)$.
\end{theorem}

\begin{proof}
We break the payment $\inspayi(\vals)$ for agent $i$ into two parts,
the part where $\indvali = \vali$ and the part where $\indvali \neq
\vali$.  In the latter part the payment identity requires zero
payment.  In the former, which happens with probability $1-\insrate$,
the payment conditioned on $\indvalsmi$ is $\epvpayi(\vali,\indvalsmi)
= \vali\, \epvalloci(\vali,\indvalsmi) - \int_0^{\vali}
\epvalloci(z,\indvalsmi)\,dz$.  The expected payment conditioned on
$\indvalsmi$ (but not conditioning on $\indvali = \vali$) is $(1-\insrate)\,
\epvpayi(\vali,\indvalsmi)$.  

Now evaluate $\expect{\indinspayi \given \indvalsmi}$ as defined in
\Cref{d:instrumented-scf} as follows:
\begin{align*}
\expect{\indinspayi \given \indvalsmi} &=
  \vali \, \expect{\indalloci \given \indvalsmi} - \vali\,\tfrac{1-\insrate}{\insrate}\, \expect{\indbelowi \given \indvalsmi}\\
&= \vali\, (1-\insrate)\, \epvalloci(\vali,\indvalsmi)
  - \vali\, \tfrac{1-\insrate}{\insrate}\,\tfrac{\insrate}{\vmax}\,\tfrac{\vmax}{\vali} \,\int_0^{\vali} \epvalloci(z,\indvalsmi)\, dz\\
&= (1-\insrate) \, \epvpayi(\vali,\indvalsmi). 
\end{align*}

Since the expected payments conditioned on $\indvalsmi$ are equal, so are the
unconditional expected payments.
\end{proof}

Instrumentation allows the mechanism to construct a consistent
estimator for the realized allocation rule that is otherwise unknown.
This estimator can be used to construct a dashboard.

\begin{definition}\label{d:empirical-dashboard}
For stage $t \geq t_0$, the {\em instrumented
  dashboard} is the profile of single-agent bid allocation rules
$\dallocs \super {t}$ defined by $\dalloci \super {t}$ for agent $i$ as follows:
\begin{enumerate}
\item[0.] Let $\evali \super s$ denote the agent's inferred value in stage $s \in \{1, \ldots, t-1\}$.
\item Define allocation data set corresponding to the uniform instrumentation $\{(\indvali \super s, \indinsalloci \super s) : s \in \{1,\ldots,t-1\} \wedge \indvali \neq \evali \super s \}$.
\item Define the empirical average allocation as $\empavgi \super t$ as the
  average allocation of this data set.
\item Define the empirical allocation rule as $\empalloci \super t : \reals \to [0,1]$ as a continuous isotonic regression of the allocation data set.
\item Define the instrumented allocation rule as $\ealloci \super t = (1-\insrate)\,\empalloci \super t(\val) + \insrate\, \empavgi \super t$.
\item The instrumented dashboard for agent $i$ is the
  bid-allocation rule $\dalloci \super t$ that corresponds to
  $\ealloci \super t$ via
  equation~\eqref{eq:single-agent-bid-allocation-rule}.
\end{enumerate}
\end{definition}

Standard methods for isotonic regression can be used to estimate the
empirical allocation rule in \Cref{d:empirical-dashboard}.  Errors in
the regression will be resolved by the rebalancing approach of
\Cref{s:rebalancing}.  A key required property, however, is that
dashboard is continuous.  Thus, isotonic regressions that result in
continuous functions should be preferred.  Approaches to isotonic
regression in statistics, including smoothing and penalization (see
\citealp{ramsay:88}, \citealp{mammen:91}, \citealp{kakade:11})
guarantee continuity of the resulting estimator for the regression
function.  The ironing procedure common in Bayesian mechanism design
is an isotonic regression; however, it results in discontinuous
functions and is, thus, inappropriate for constructing dashboards.

The analysis below focuses on the winner-pays-bid rebalancing
instrumented dashboard for the instrumented allocation algorithm.
(Recall that all-pay mechanisms have deterministic payments and, thus,
the budget imbalance generated at a given round can be resolved
deterministically.  We omit the simple analysis.)  With
winner-pays-bid mechanisms the payments are only made when the agent
is allocated.  As a result, the payment residual and balanced resolved
are stochastic and necessitate a more sophisticated analysis.

The rebalancing approach in \Cref{s:rebalancing} used the functional
form of the allocation algorithm to calculate the difference between
actual payments given by the payment format and incentive compatible
payments given by the payment identity.  This payment residual is added to a
balance and each stage a portion of the balance is added to the
payment of a type with zero value when determining the dashboard.  Our
approach here is to instead use the implicit payments which, by
\Cref{thm:unbiased}, are unbiased estimators of the incentive
compatible payments.  The three terms in the balance update formula
below are the previous balance, the implicit payment of the current
stage, and the actual payment of the current stage.

\begin{definition}
From stage $s$ with outstanding balance $\indcumbal \super s$, agent
bid $\bid \super s$, realized instrumented allocation $\indinsalloc
\super s$, and realized implicit payment $\indinspay \super s$; the stage $s+1$
outstanding balance in the winner-pays-bid instrumented rebalancing
dashboard is:
$$
\indcumbal \super {s+1}=\indcumbal \super s +
\indinspay \super s - \indinsalloc \super s\, \bid \super s.
$$
\end{definition}

We first confirm that in expectation the analysis of
\Cref{s:rebalancing} holds.  Recall from
equation~\eqref{eq:rebalancing} that the bid strategy for the payment
rebalancing dashboard for allocation rule $\ealloc$ is
$\strat\primed(\val) = \strat(\val) + \indcumbal\, \rebalrate /
\ealloc(\val)$.  Thus the actual payment in the balance update formula
can be split for inferred value $\eval$ satisfying $\bid =
\strat\primed(\eval)$ giving a balance update as
\begin{align*}
\indinspay - \indinsalloc \, \bid &= \indinspay - \indinsalloc \, \strat(\eval) - \indinsalloc\,\indcumbal \, \rebalrate / {\ealloc(\eval)}.
\\
\intertext{Define the {\em payment residual} (cf. \Cref{d:payment-residual}) as the first two terms and the {\em resolved balance} as the last term:}
\indnewbal &= \indinspay - \indinsalloc \, \strat(\eval), &
\indrebal &= \indinsalloc\,\indcumbal \, \rebalrate/ {\ealloc (\eval)}.\\
\intertext{Thus, the total change to the cumulative balance is $\indnewbal - \indrebal$.  Importantly the expected payment residual $\expect{\indnewbal}$ matches \Cref{d:payment-residual}.  Taking expectations and applying \Cref{thm:unbiased} 
%and noting that $\ealloc(\eval)\,\strat(\eval)$ is the incentive compatible payment $\vpay(\eval)$ for $\ealloc$; 
we have,}
\expect{\indnewbal} & =
\inspay(\eval) - \insvalloc(\eval) \, \strat(\eval) 
\end{align*}
Here $\inspay$ is the payment rule for $\insvalloc$ and payment
$\ealloc(\eval)\, \strat(\eval)$ satisfies the payment identity for
$\ealloc$.  Thus, if the dashboard is correct, i.e., $\ealloc =
\insvalloc$, then the expected payment residual $\expect{\indnewbal}$
is zero.  When we have incorrect estimates, the extent to which the
difference of the first terms is not zero is gives a payment residual
that must be rebalanced in the future.

%% \begin{theorem} 
%% \label{t:blackbox-expected-rebalancing}
%% For any monotonic instrumented stage allocation rules $\insalloc
%% \super 1,\ldots, \insalloc \super t$, any monotonic dashboard rules
%% $\ealloc \super 1, \ldots, \ealloc \super t$, and any inferred values
%% $\eval \super 1,\ldots,\eval \super t$; the single-call
%% winner-pays-bid payment rebalancing dashboard has total balance at
%% stage $t$ equal to that of the functional winner-pays-bid payment
%% rebalancing dashboard for these allocation rules and with the same
%% parameters.
%% \end{theorem}
%% \begin{corollary}
%% \label{c:blackbox-expected-rebalancing}
%% In the environment of \Cref{t:blackbox-expected-rebalancing} with
%% allocation probabilities supported on $[\minalloc,1]$, rebalancing
%% rate $\rebalrate \in (0,\minalloc)$, and per-stage generated imbalance at
%% most $\maxbal$, the total balance at stage $t$ of at most
%% $\maxbal / \rebalrate$.
%% \end{corollary}

%% While the equivalence of \Cref{t:blackbox-expected-rebalancing} and
%% bound of \Cref{c:blackbox-expected-rebalancing} hold in expectation,
%% in the remainder of this section we derive high probability bounds.

Our analysis starts with three observations:
\begin{itemize}
\item When $\indinsalloc = 0$ the payment residual and amount
  rebalanced are zero, i.e., $\indnewbal = \indrebal = 0$; otherwise:
\item the size of the range of $\indnewbal$ is $\eval/\insrate$ (equal
  to the size of the range of $\indinspay$ which equals $\eval\,[-
    (1-\insrate) / \insrate, 1] $ for inferred value $\eval$);
  and
\item the amount rebalanced is $\indcumbal \rebalrate /
  \ealloc(\eval)$ for inferred value $\eval$.
\end{itemize}
Our analysis proceeds like that of \Cref{s:rebalancing}.  We consider
the $\numrebalstages$ stages $s$ where the agent is allocated, i.e.,
$\rebalstages = \{s : \indinsalloc \super s \neq 0\}$.  These are the
stages where $\indnewbal \super s$ and $\indrebal \super s$ can be
non-zero.  We consider the amount of the payment residual $\indnewbal
\super s$ from stage $s \in \rebalstages$ that remains at final stage
$t$ when, at each subsequent stage in $\rebalstages$, a fraction of that
payment residual is resolved.
\Cref{l:remaining-imbalance} shows that the balance resolved in each
stage $s \in \rebalstages$ is between $\cumbal \super s\, \rebalrate$
and $\cumbal \super s$.  The magnitude of the payment residual is
upper bounded by $\eval / \insrate$ in the lemma below.

\begin{lemma}
\label{l:single-call-maximum-imbalance}
In a stage in which the agent's inferred value is $\eval$, the
single-call winner-pays-bid rebalancing dashboard for instrumented
allocation algorithm with instrumentation probability $\insrate$ has
payment residual with magnitude $|\indnewbal|$ at most
$\eval/\insrate$.
\end{lemma}
\begin{proof}
When $\indalloc = 0$ the payment residual is zero.  The payment
residual when $\indalloc =1$ is $\indnewbal = \indinspay -
\strat(\eval)$ with $\indinspay = \eval\, (\indalloc -
\tfrac{1-\insrate}{\insrate}\, \indbelow)$.  The magnitude $|\indnewbal|
\leq \tfrac{1-\insrate}{\insrate} \eval + \strat(\eval)$.  Since the
winner-pays-bid strategy (without the addition for rebalancing) is
individually rational $\strat(\eval) \leq \eval$ and, thus, $|\indnewbal|
\leq \eval / \insrate$.
\end{proof}

We are ready now to apply \Cref{t:cumulative-imbalance} to the
rebalancing dashboard for the instrumented allocation algorithm.  A
helpful property of the instrumented allocation algorithm
(\Cref{d:instrumented-scf}) is that the instrumentation implies that
the allocation probabilities for all agents are bounded away from zero,
i.e., the minimum allocation probability for an agent $i$ in the
dashboard is $\insrate$ times the average allocation probability for
$i$ with value $\indvali \sim U[0,\vmax]$ in the induced allocation
algorithm.  We need to set $\rebalrate$ to lower bound this quantity.

\begin{corollary}
\label{c:single-call-cumulative-imbalance}
For the single-call winner-pays-bid rebalancing dashboard for
instrumented allocation algorithm with rebalancing rate $\rebalrate$
(set appropriately), instrumentation parameter $\insrate$, and values
in $[0,\vmax]$; the outstanding balance at stage $t$ is at most $\vmax /
(\insrate\,\rebalrate)$.
\end{corollary}

Note that this bound is on the realized total balance and is not a
high-probability or in-expectation result; it holds always.  The
payment residual at any stage is a random variable with expected
magnitude at most $\vmax$ and range at most $\vmax/\insrate$.  The
expected payment residual, however, has magnitude at most $\vmax$
(\Cref{l:maximum-imbalance}).  \Cref{t:cumulative-imbalance} implies
that the expected outstanding balance at time $t$ is at most $\vmax /
\rebalrate$.  Incentive consistency is defined in expectation over
randomization in the mechanism and strategies.  We have the following
corollary.  As before, holding $\vmax/\rebalrate$ constant, the
incentive inconsistency vanishes with $t$.

\begin{corollary}
\label{c:single-call-incentive-consistency}
For the single-call winner-pays-bid rebalancing dashboard for
instrumented allocation algorithm with rebalancing rate $\rebalrate$
(set appropriately), instrumentation parameter $\insrate$, and values
in $[0,\vmax]$; over $t$ stages the dashboard mechanism is
$\vmax/(\rebalrate\, t)$ incentive inconsistent for all strategies.
\end{corollary}

The bound of \Cref{c:single-call-cumulative-imbalance} can be improved
via the Chernoff-Hoeffding inequality.  We show that the outstanding
balance at time $t$ is the weighted average of these payment residuals
with weights that are geometrically decreasing.  This results in the
following high-probability bound.

\begin{theorem}
\label{t:cumulative-imbalance:stochastic}
For the single-call winner-pays-bid rebalancing dashboard for
instrumented allocation algorithm with rebalancing rate $\rebalrate$
(set appropriately), instrumentation parameter $\insrate$, and values
in $[0,\vmax]$; the outstanding balance at stage $t$ is at most $\vmax
/ \rebalrate +
\frac{\vmax}{\insrate}\sqrt{\frac{1}{2\eta}\log\left(\frac{2}{\delta}\right)}$
with probability at least $1-\delta$.
\end{theorem}

\begin{proof}
By \Cref{l:single-call-maximum-imbalance} for each stage $s$,
$|\indnewbal \super s| \leq \vmax/\insrate.$ At the same time
$\expect{\indnewbal \super s} \leq \vmax$ and $\indnewbal \super s=0$
whenever $\indinsalloc \super s=0$. As in the proof of
\Cref{t:cumulative-imbalance}, consider the stages $\rebalstages =
 \{s_0,\ldots,s_{\numrebalstages-1}\}$ with non-zero payments indexed
in decreasing order.
Then
$$
\indcumbal \super t - \expect{\indcumbal \super t}=\sum\limits^{\numrebalstages}_{k=0}\prod_{l=0}^{k-1}\left(1-
\rebalrate/\ealloci\super {s_l} (\eval \super {s_l})
\right)\left(\indnewbal \super {s_k} - \expect{\indnewbal \super {s_k}}\right),
$$
where the range of each term in the sum is bounded by $(1-\eta)^k\,\vmax /\insrate.$ 

Then we can apply the Chernoff-Hoeffding inequality
to evaluate
\begin{align*}
    \prob{
    \left|\indcumbal \super t - \expect{\indcumbal \super t}\right| > \xi \,\Big|\, \indinsalloc \super 1,\ldots, \indinsalloc \super t
    } \leq 2 \exp\left(
    -\frac{2 \xi^2}{\sum\nolimits^{\infty}_{l=0}
    (1-\rebalrate)^{2l}\vmax^2 /\insrate^2}
    \right) \leq 2 \exp \left(-\frac{2\xi^2 \insrate^2 \rebalrate}{\vmax^2}
    \right)
\end{align*}
which follows from 
$
\left(\sum\nolimits^{\numrebalstages}_{l=0}
    (1-\rebalrate)^{2l}\right)^{-1} \geq \left(\sum\nolimits^{\infty}_{l=0}
    (1-\rebalrate)^{2l}\right)^{-1} =1-(1-\rebalrate)^2 \geq \rebalrate ,
$
since $0<\rebalrate<1$ and $\rebalrate \geq \rebalrate^2.$ 

Then with probability $\delta$ the imbalance 
of empirical dashboard with instrumentation can deviate from the expectation 
no further than 
$
\frac{\vmax}{\insrate}\sqrt{\frac{1}{
2\eta}\log\left(\frac{2}{\delta}\right)}.
$
\end{proof}

%% file: main.bbl
\begin{thebibliography}{}

\bibitem[Archer and Tardos, 2001]{AT-01}
Archer, A. and Tardos, {\'E}. (2001).
\newblock Truthful mechanisms for one-parameter agents.
\newblock In {\em Proceedings 2001 IEEE International Conference on Cluster
  Computing}, pages 482--491. IEEE.

\bibitem[Armstrong, 1999]{arm-99}
Armstrong, M. (1999).
\newblock Price discrimination by a many-product firm.
\newblock {\em The Review of Economic Studies}, 66(1):151--168.

\bibitem[Athey and Nekipelov, 2010]{AN-10}
Athey, S. and Nekipelov, D. (2010).
\newblock A structural model of sponsored search advertising auctions.
\newblock In {\em Sixth ad auctions workshop}.

\bibitem[Babaioff et~al., 2014]{BILW-14}
Babaioff, M., Immorlica, N., Lucier, B., and Weinberg, S.~M. (2014).
\newblock A simple and approximately optimal mechanism for an additive buyer.
\newblock In {\em 2014 IEEE 55th Annual Symposium on Foundations of Computer
  Science}, pages 21--30. IEEE.

\bibitem[Babaioff et~al., 2010]{BKS-10}
Babaioff, M., Kleinberg, R.~D., and Slivkins, A. (2010).
\newblock Truthful mechanisms with implicit payment computation.
\newblock In {\em Proceedings of the 11th ACM conference on Electronic
  commerce}, pages 43--52. ACM.

\bibitem[Bei and Huang, 2011]{BH-11}
Bei, X. and Huang, Z. (2011).
\newblock Bayesian incentive compatibility via fractional assignments.
\newblock In {\em Proceedings of the twenty-second annual ACM-SIAM symposium on
  Discrete Algorithms}, pages 720--733. Society for Industrial and Applied
  Mathematics.

\bibitem[Cai et~al., 2013a]{CDW-13a}
Cai, Y., Daskalakis, C., and Weinberg, S.~M. (2013a).
\newblock Reducing revenue to welfare maximization: Approximation algorithms
  and other generalizations.
\newblock In {\em Proceedings of the twenty-fourth annual ACM-SIAM symposium on
  Discrete algorithms}, pages 578--595. Society for Industrial and Applied
  Mathematics.

\bibitem[Cai et~al., 2013b]{CDW-13b}
Cai, Y., Daskalakis, C., and Weinberg, S.~M. (2013b).
\newblock Understanding incentives: Mechanism design becomes algorithm design.
\newblock In {\em 2013 IEEE 54th Annual Symposium on Foundations of Computer
  Science}, pages 618--627. IEEE.

\bibitem[Caragiannis et~al., 2011]{CKKK-11}
Caragiannis, I., Kaklamanis, C., Kanellopoulos, P., and Kyropoulou, M. (2011).
\newblock On the efficiency of equilibria in generalized second price auctions.
\newblock In {\em Proceedings of the 12th ACM conference on Electronic
  commerce}, pages 81--90. ACM.

\bibitem[Dughmi et~al., 2017]{DHKN-17}
Dughmi, S., Hartline, J.~D., Kleinberg, R., and Niazadeh, R. (2017).
\newblock Bernoulli factories and black-box reductions in mechanism design.
\newblock In {\em Proceedings of the 49th Annual ACM SIGACT Symposium on Theory
  of Computing}, pages 158--169. ACM.

\bibitem[D{\"u}tting and Kesselheim, 2015]{DK-15}
D{\"u}tting, P. and Kesselheim, T. (2015).
\newblock Algorithms against anarchy: Understanding non-truthful mechanisms.
\newblock In {\em Proceedings of the Sixteenth ACM Conference on Economics and
  Computation}, pages 239--255. ACM.

\bibitem[Gorokh et~al., 2017]{GBI-17}
Gorokh, A., Banerjee, S., and Iyer, K. (2017).
\newblock From monetary to non-monetary mechanism design via artificial
  currencies.
\newblock In {\em Proceedings of the 2017 ACM Conference on Economics and
  Computation}, pages 563--564. ACM.

\bibitem[Hartline and Taggart, 2016]{HT-16}
Hartline, J. and Taggart, S. (2016).
\newblock Non-revelation mechanism design.
\newblock {\em arXiv preprint arXiv:1608.01875}.

\bibitem[Hartline and Taggart, 2019]{HT-19}
Hartline, J. and Taggart, S. (2019).
\newblock Sample complexity for non-truthful mechanisms.
\newblock In {\em Proceedings of the 2019 ACM Conference on Economics and
  Computation}, pages 399--416. ACM.

\bibitem[Hartline et~al., 2015]{HKM-15}
Hartline, J.~D., Kleinberg, R., and Malekian, A. (2015).
\newblock Bayesian incentive compatibility via matchings.
\newblock {\em Games and Economic Behavior}, 92:401--429.

\bibitem[Hartline and Lucier, 2010]{HL-10}
Hartline, J.~D. and Lucier, B. (2010).
\newblock Bayesian algorithmic mechanism design.
\newblock In {\em Proceedings of the forty-second ACM symposium on Theory of
  computing}, pages 301--310. ACM.

\bibitem[Hartline and Lucier, 2015]{HL-15}
Hartline, J.~D. and Lucier, B. (2015).
\newblock Non-optimal mechanism design.
\newblock {\em American Economic Review}, 105(10):3102--24.

\bibitem[Jackson and Sonnenschein, 2007]{JS-07}
Jackson, M.~O. and Sonnenschein, H.~F. (2007).
\newblock Overcoming incentive constraints by linking decisions 1.
\newblock {\em Econometrica}, 75(1):241--257.

\bibitem[Johari and Tsitsiklis, 2004]{JT-04}
Johari, R. and Tsitsiklis, J.~N. (2004).
\newblock Efficiency loss in a network resource allocation game.
\newblock {\em Mathematics of Operations Research}, 29(3):407--435.

\bibitem[Kakade et~al., 2011]{kakade:11}
Kakade, S.~M., Kanade, V., Shamir, O., and Kalai, A. (2011).
\newblock Efficient learning of generalized linear and single index models with
  isotonic regression.
\newblock In {\em Advances in Neural Information Processing Systems}, pages
  927--935.

\bibitem[Leme and Tardos, 2010]{LT-10}
Leme, R.~P. and Tardos, E. (2010).
\newblock Pure and {B}ayes-{N}ash price of anarchy for generalized second price
  auction.
\newblock In {\em Foundations of Computer Science (FOCS), 2010 51st Annual IEEE
  Symposium on}, pages 735--744. IEEE.

\bibitem[Lucier and Borodin, 2010]{LB-10}
Lucier, B. and Borodin, A. (2010).
\newblock Price of anarchy for greedy auctions.
\newblock In {\em Proceedings of the twenty-first annual ACM-SIAM symposium on
  Discrete Algorithms}, pages 537--553. Society for Industrial and Applied
  Mathematics.

\bibitem[Mammen, 1991]{mammen:91}
Mammen, E. (1991).
\newblock Estimating a smooth monotone regression function.
\newblock {\em The Annals of Statistics}, pages 724--740.

\bibitem[Myerson, 1981]{mye-81}
Myerson, R.~B. (1981).
\newblock Optimal auction design.
\newblock {\em Mathematics of Operations Research}, 6(1):58--73.

\bibitem[Myerson and Satterthwaite, 1983]{MS-83}
Myerson, R.~B. and Satterthwaite, M.~A. (1983).
\newblock Efficient mechanisms for bilateral trading.
\newblock {\em Journal of economic theory}, 29(2):265--281.

\bibitem[Ramsay, 1988]{ramsay:88}
Ramsay, J. (1988).
\newblock Monotone regression splines in action.
\newblock {\em Statistical science}, 3(4):425--441.

\bibitem[Syrgkanis and Tardos, 2013]{ST-13}
Syrgkanis, V. and Tardos, E. (2013).
\newblock Composable and efficient mechanisms.
\newblock In {\em Proceedings of the forty-fifth annual ACM symposium on Theory
  of computing}, pages 211--220. ACM.

\bibitem[Wilkens and Sivan, 2015]{WS-15}
Wilkens, C.~A. and Sivan, B. (2015).
\newblock Single-call mechanisms.
\newblock {\em ACM Transactions on Economics and Computation (TEAC)}, 3(2):10.

\end{thebibliography}
